\documentclass[10pt,journal]{IEEEtran}

\usepackage{amsmath}
\allowdisplaybreaks
\usepackage{cite}
\usepackage{graphicx}
\usepackage{textcomp}
\usepackage{xcolor} 
\usepackage{enumitem}

\usepackage{multirow}
\usepackage{subcaption}
\usepackage{amssymb,amsfonts}
\usepackage{amsthm}
\usepackage{bm}
\usepackage[colorlinks=true, linkcolor=blue, citecolor=blue, urlcolor=blue]{hyperref}
\usepackage{algorithm}
\usepackage{algpseudocode}
\usepackage{amssymb}

\newtheorem{assumption}{Assumption}
\newtheorem{definition}{Definition}
\newtheorem{lemma}{Lemma}
\newtheorem{proposition}{Proposition}
\newtheorem{corollary}{Corollary}
\newtheorem{remark}{Remark}

\begin{document}
\bstctlcite{BSTcontrolNoDash}
\title{Optimal Movable-Antenna Control for \\Multi-Path Sensing Guided by Prior AoA Statistics}

\author{Jaehong~Kim,~\IEEEmembership{Graduate Student Member,~IEEE}, Changsheng~You,~\IEEEmembership{Member,~IEEE}, \\ Jihong~Park,~\IEEEmembership{Senior Member,~IEEE}
        and Seung\mbox{-}Woo~Ko,~\IEEEmembership{Senior Member,~IEEE}%
        \vspace{-0pt}
\thanks{%
J. Kim and S.-W. Ko are with Inha University, Incheon 22212, Republic of Korea (e-mail: kimjaehong@inha.edu, swko@inha.ac.kr). J. Park is with Singapore University of Technology and Design, Singapore (e-mail: jihong\_park@sutd.edu.sg). C. You is with Southern University of Science and Technology, Shenzhen, China (e-mail: youcs@sustech.edu.cn).

An earlier version of this paper will be presented in part at the
IEEE Global Communications Conference (GLOBECOM) 2026, and is available on arXiv~\cite{Kim2026Arxiv}.
%
\vspace{0pt}
}

}


\maketitle

\begin{abstract}
Multi-path sensing, which aims to extract the geometric attributes of multiple propagation paths, is expected to be a key functionality of 6G. A \emph{movable antenna} (MA) can enable this functionality by synthesizing an aperture through mechanical motion. However, existing MA-based sensing methods typically rely on exhaustive scanning over the entire movable region, resulting in significant control overhead and sensing latency, which limit their practicality for agile sensing. To address this challenge, this paper develops a prior-guided agile multi-path sensing framework that leverages weak prior \emph{angle-of-arrival} (AoA) statistics as side information. The proposed framework is built on two key steps. First, the movable plate's three-dimensional orientation is optimized only once to configure a mechanically feasible scan region that enhances path visibility while preserving inter-path discriminability, guided by Fisher information analysis. Second, given the optimal plate orientation, the MA performs only two linear scans, whose non-collinear spatial phase projections are fused with the prior AoA statistics through a \emph{maximum a posteriori} (MAP)-based estimator to recover the elevation and azimuth AoAs of multiple paths. The estimated AoAs are subsequently used to extract the \emph{times-of-arrival} (ToAs) by enhancing the target path component while suppressing interference from other paths. With only one orientation adjustment and two linear scans, the proposed framework enables agile multi-path sensing with significantly reduced control overhead and latency, while achieving AoA and ToA estimation accuracy close to the single-path benchmark.
\end{abstract}
\begin{IEEEkeywords}
Multi-path sensing, movable antenna, orientation control, linear scan, angle-of-arrival, time-of-arrival.   
\end{IEEEkeywords}

\section{Introduction}

One key vision of 6G systems is to exploit radio signals not only for communication but also for high-precision sensing \cite{DZhang2026}. Unlike conventional sensing modalities such as LiDAR and cameras, radio waves naturally propagate through multiple paths. This intrinsic multi-path behavior enables the detection of objects and the inference of environmental geometry even under \emph{Non-Line-of-Sight} (NLoS) conditions \cite{BaqueroBarneto2022}. In this work, we refer to the process of extracting the geometric attributes of these propagation paths, specifically their \emph{Angles-of-Arrival} (AoAs) and \emph{Times-of-Arrival} (ToAs), without performing full channel estimation, as \emph{multi-path sensing}. A promising approach to enhance multi-path sensing capability is to employ a \emph{movable antenna} (MA) \cite{Zhu2026TutorialMA} at the receiver. Through controlled mechanical repositioning over a small spatial region, the MA creates a synthetic aperture, allowing the receiver to capture richer spatial information and to separate closely spaced \emph{signal paths} (SPs) when its motion is appropriately controlled. Motivated by this capability, we propose a prior-aided MA control framework that leverages weak prior AoA statistics to determine a mechanically feasible scan region for the MA. 
Within the configured region, the MA performs only two linear scans to estimate the AoAs and ToAs of multiple SPs, thereby significantly reducing mechanical overhead while maintaining high sensing~resolution.

\vspace{-10pt}
\subsection{Prior Work}

In this subsection, we briefly review existing wireless sensing methodologies, summarizing their representative approaches and inherent limitations.

\subsubsection{6G Sensing}

6G sensing has been commonly categorized into two paradigms, namely sensing-assisted communication and communication-assisted sensing~\cite{FLiu2022}. The former exploits SP parameters such as AoA and ToA for beam management~\cite{Xue2024BeamManagement}, whereas the latter, which is the focus of this work, leverages communication signals to accurately estimate SP parameters~\cite{Gonzalez2024}. Although SP parameters can be reliably estimated in single-path scenarios, the estimation performance significantly degrades in multi-path environments. In particular, different SPs become difficult to resolve when their geometric parameters, such as AoAs or ToAs, are closely~spaced. 

A straightforward way to mitigate this ambiguity is to enlarge the receiver's spatial aperture. In this regard, \emph{extremely large-scale MIMO} (XL-MIMO) has been considered a promising solution, as its widely distributed antenna elements can capture rich spatial variations of the received signal~\cite{Li2019MassiveMIMOLocalization}. This spatial diversity helps resolve closely spaced SPs and improves SP parameter estimation. However, these gains come at the cost of deploying a very large number of antenna elements and associated active RF components, leading to substantial hardware, power, calibration, and deployment costs. Consequently, network operators may be reluctant to adopt XL-MIMO as a dense, large-scale sensing infrastructure~\cite{Bjornson2014HardwareScaling,Gao2016LowRFComplexity}.

Recent efforts have explored more cost-effective ways to enhance SP separability, among which \emph{reconfigurable intelligent surfaces} (RISs) and MAs are representative examples. Both aim to improve sensing performance by making SPs more distinguishable at the receiver, rather than simply increasing the number of active antenna elements. Their operating principles, however, are different. An RIS modifies the propagation of reflected SPs by tuning the phase responses of its passive elements~\cite{Wu2024IntelligentSurfaces}. By steering or reshaping reflected SPs toward desired directions, RISs can improve the separability of SPs and thereby enhance SP parameter estimation. Nevertheless, an RIS can manipulate only the SPs that physically impinge on its surface. To cover multiple SPs of interest, an RIS should therefore either be made extremely large~\cite{Kang2025NearFieldRISLocalization} or be deployed at locations where the SPs are likely to be reflected~\cite{Alexandropoulos2022MultiRISLocalization}, both of which constrain practical deployment.

In contrast, an MA enhances signal-path separability directly at the receiver by mechanically repositioning a single antenna element within its movable plate~\cite{Zhu2026TutorialMA}. This allows the receiver to synthesize a large effective aperture, as in XL-MIMO, while favorably adjusting the spatial sampling of multiple SPs through controlled antenna motion. Therefore, MAs provide a cost-effective and deployment-friendly way to obtain spatial sensing diversity without requiring many active RF chains or carefully placed external reflecting surfaces.

\subsubsection{Movable Antenna-Based Sensing}
Several works on MA-based sensing have been proposed in the literature, grouped into model-based and learning-based approaches.

First, model-based approaches leverage explicit signal models of the sequence of received signals collected during the MA scanning process. By casting these received signals as a well-defined mathematical problem, these approaches enable us to use optimization techniques that provide tractable algorithms and theoretical performance guarantees for SP parameter sensing. A typical example is compressed sensing \cite{Donoho2006CompressedSensing}, which exploits the sparse nature of multi-path propagation and formulates SP parameter estimation as a sparse recovery problem. In~\cite{MA2023}, a successive transmitter-receiver compressed sensing technique is employed to estimate the SP parameters by alternately scanning the movable regions at the transmitter and receiver. In~\cite{ZXiao2024}, an \emph{orthogonal matching pursuit}~(OMP) technique is applied by correlating the received signals with a position-dependent dictionary constructed from the MA sampling positions and progressively canceling the already detected path components from the received signals. This idea is further extended to wideband systems in~\cite{SCao2025}, which develops a \emph{simultaneous OMP}~(SOMP)-based framework that jointly exploits the common sparsity shared across subcarriers. For the near-field regime, a subregion-based estimation technique is developed in~\cite{Sun2026NearFieldMA}, where the movable region is partitioned into subregions in which the far-field approximation holds, the per-subregion angles are estimated via {Newtonized OMP}, and the scatterers are localized by clustering the resulting directional rays across subregions.
Another representative technique is tensor decomposition, which exploits the low-rank structure of the received signals \cite{Sidiropoulos2017TensorSP}. In \cite{RZhang2024TensorMA}, the received signals collected over multiple MA positions are expressed as third-order tensors, enabling SP parameter estimation via canonical polyadic decomposition.  

Second, learning-based approaches employ neural networks for multi-path sensing. By learning a direct mapping from sequences of received signals to SP parameters using training data, these approaches can exploit latent statistical patterns embedded in MA measurements and, in some cases, jointly design the MA sampling positions and the neural estimator. Once trained, they can estimate the SP parameters with low online computational complexity, without exhaustively searching over candidate antenna positions. In~\cite{Jang2025LearningMA}, for example, the MA sampling positions and the channel-angle estimator are jointly optimized offline using a neural network trained over a statistical distribution of channel components. During online estimation, the learned positions remain fixed, while the remaining network layers estimate the \emph{Angles-of-Departure} (AoDs) and AoAs from the received signal. This framework is extended to wideband systems in~\cite{Jang2026DLMAWideband}, where the ToA of each SP is jointly estimated and paired with its corresponding AoD and AoA.

More recently, hybrid approaches that integrate model-based and learning-based ones have emerged to combine their complementary strengths. In~\cite{Feng2026DLMA}, the SP parameters are first estimated using SOMP~\cite{SCao2025}, after which the resulting model-based estimates are refined by a Swin-Transformer-based denoising network~\cite{Liu2021SwinTransformer}. This two-stage framework preserves the structural interpretability and tractability of compressed sensing while leveraging the noise-suppression and nonlinear-representation capabilities of neural networks.

Despite their effectiveness, the aforementioned approaches incur substantial overhead. Model-based approaches typically require measurements from many MA positions to capture sufficient spatial information for accurate SP parameter estimation. Learning-based approaches can reduce online search complexity, but they require large, representative training datasets and may need retraining when the propagation statistics change. These requirements limit sensing agility and increase the overall implementation cost~\cite{Zhu2024MovableAntennasMagazine}.

\vspace{-10pt}
\subsection{Contributions}
This work considers a model-based approach in which the receiver is equipped with an MA mounted on a movable plate. As discussed above, multi-path sensing requires the MA to be scanned over a large number of positions across the entire plate to collect sufficient spatial information, which limits sensing agility. To overcome this limitation, we leverage prior AoA statistics inferred from the surrounding environment, such as nearby static buildings and obstacles. Such statistics can be readily obtained using several techniques, including large multimodal models~\cite{Kim2026} and digital twins~\cite{Abouamer2025}. We treat these prior statistics as weak yet useful side information and use them to optimize the movable plate's orientation before scanning. Guided by these priors, the plate is configured only once so that all SPs remain visible and their SP parameters well separated. As a result, the AoA pairs and ToAs of the multiple SPs can be rapidly estimated from only two linear MA scans. The main contributions are as follows:
\begin{itemize}[leftmargin=*]
\item \textbf{Movable-Plate Orientation Control:}
We first develop a prior-driven movable-plate orientation-control scheme that optimizes the three-dimensional orientation of the movable plate before scanning using only prior AoA statistics. From a Fisher information perspective, the AoA of each SP can be estimated more accurately when the corresponding SP projections onto the plate are more clearly separated. 
To promote this separability under AoA uncertainty, we characterize the probability that the projected coordinates of any SP pair reverse the ordering implied by their mean projections, referred to as an \emph{order-reversal event}. We derive a tractable closed-form upper bound on this probability using a probabilistic inequality (see Proposition \ref{prop:cantelli}) and incorporate it into the orientation-design objective. Next, we impose a front-side constraint to ensure that all SPs impinge on the front side of the plate (see Proposition \ref{prop:front-side-condition}). The resulting optimization problem is solved using \emph{sequential quadratic programming} (SQP) \cite{SQP2000}. The optimized orientation thereby makes the SP projections reliably separable, facilitating accurate estimation of their parameters.

\item \textbf{AoA and ToA Estimation with Two Linear Scans:} Given the optimized orientation, the MA performs two linear scans along the horizontal and vertical directions of the plate. Each scan yields a set of AoA-related projection parameters, which should be correctly paired across the two scans to recover the azimuth and elevation AoAs of each SP. 
To this end, we formulate a \emph{maximum a posteriori}~(MAP)-based pairing rule that incorporates the prior AoA statistics and develop a low-complexity algorithm that evaluates this rule only for a small set of likely candidate pairings. Using the recovered AoAs, the ToA of each SP is then estimated by applying spatial filtering to the received signals collected during the two scans, thereby enhancing the resolution of the target SP's ToA while suppressing the others.

\item \textbf{Validation in Realistic Ray-Traced Environments:}
We validate the proposed framework using NVIDIA Sionna RT in ray-traced models of the Florence Duomo and the Inha Aerospace Campus, accounting for reflection, diffraction, and scattering~\cite{Hoydis2023SionnaRT}. The results confirm that realistic environments provide informative prior AoA statistics and that the proposed algorithms remain accurate under practical multi-path propagation.
\end{itemize}

\begin{figure}[t]
  \centering
  \includegraphics[width=0.8\linewidth]{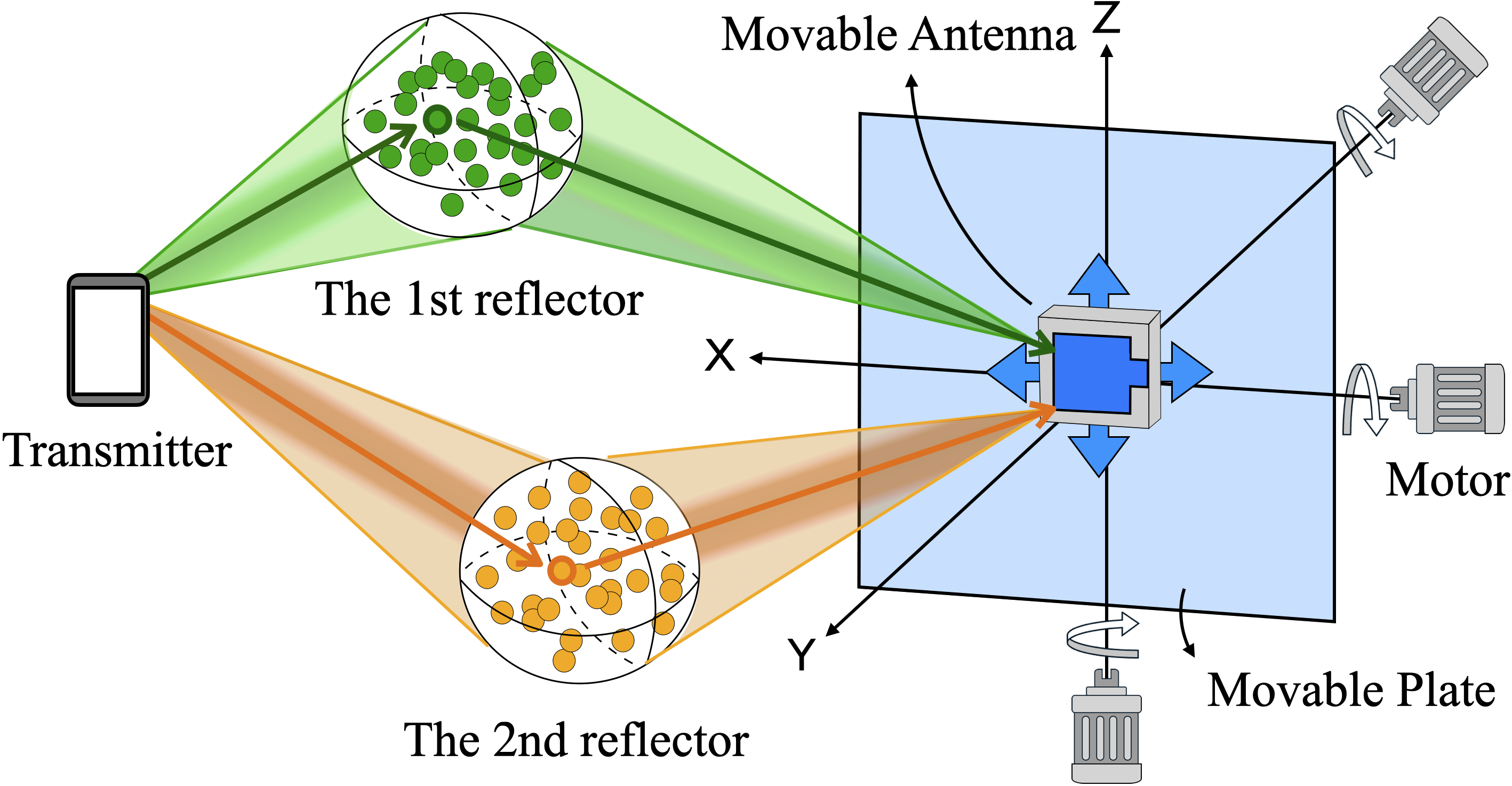}
  \caption{A graphical illustration representing MA-based multi-path sensing under AoA uncertainty ($L=2$).\vspace{-15pt}}
  \label{fig:MA Scenario}
\end{figure}

\section{System Model}\label{sec:system_model}
This section describes the system model for the multi-path sensing scenario, as shown in Fig.~\ref{fig:MA Scenario}. 
The detailed scenario is explained first, including the geometry of each entity, the parameters to be estimated, and key assumptions. Then, its signal models are described, followed by the problem definition to be addressed throughout the work. 

\vspace{-10pt}
\subsection{Scenario Description}
Consider a multi-path sensing scenario comprising a transmitter and a receiver equipped with a single rigid antenna and a single MA, respectively. The receiver's MA is mounted on a \emph{movable plate} (light blue region in Fig.~\ref{fig:MA Scenario}), whose orientation can be mechanically adjusted via a set of tilting angles to be specified in the sequel. The MA can move over this plate, thereby forming the synthetic aperture. For notational clarity, we define an initial \emph{three-dimensional} (3D) Cartesian coordinate system ($\mathsf{X}^{(0)}$, $\mathsf{Y}^{(0)}$, $\mathsf{Z}^{(0)}$) by placing the receiver's reference point at the origin, i.e., $\boldsymbol{p}_0 = [0,0,0]^\top$. In this initial frame, the $\mathsf{X}^{(0)}$-$\mathsf{Z}^{(0)}$ plane is set parallel to the width-height of the movable plate's initial orientation, while the $\mathsf{Y}^{(0)}$ axis is chosen as the plate's outward normal.

A signal broadcast from the transmitter propagates to the receiver through a multi-path environment. We consider $L$ NLoS SPs, indexed by $\ell\in\{1,\dots,L\}$. 
Each SP $\ell$ can be characterized by three parameters. The first one is the ToA, denoted by $\tau_{\ell}$, which specifies the arrival time of SP $\ell$ at the receiver's reference point $\boldsymbol{p}_0$. The second and third ones are the elevation and azimuth AoAs, denoted by $\theta_{\ell}^{(0)}$ and $\phi_\ell^{(0)}$, respectively, defined as
\begin{align}
\theta_{\ell}^{(0)}
= \cos^{-1}\!\left(\frac{[\boldsymbol{b}_\ell]_3}{\lVert \boldsymbol{b}_\ell \rVert}\right),
\quad
\phi_{\ell}^{(0)}
= {\tan}^{-1}\left(\frac{[\boldsymbol{b}_\ell]_2}{[\boldsymbol{b}_\ell]_1}\right),
\label{eq:angles}
\end{align}
where $\boldsymbol{b}_{\ell}\in\mathbb{R}^{3\times 1}$ is the last reflection point before the arrival of SP $\ell$\footnote{The transmitter's location is not explicitly required since, in a rich multi-path environment, the ToA and AoA are primarily governed by the geometry between the receiver and surrounding reflectors. On the other hand, at higher carrier frequencies (e.g., mmWave/THz) where single-bounce paths become dominant, the current framework can be extended to localizing the transmitter or mapping nearby reflectors, as in \cite{Lotti2023} and \cite{Shahmansoori2018}.}, and $[\cdot]_i$ denotes the $i$-th element of a vector. The corresponding unit direction vector of SP $\ell$, denoted by $\boldsymbol{a}_{\ell}^{(0)}$, is given as 
\begin{align}\label{eq: Initial AoA direction}
\boldsymbol{a}_\ell^{(0)} = [\sin \theta_\ell^{(0)} \cos \phi_\ell^{(0)}, \sin \theta_\ell^{(0)} \sin \phi_\ell^{(0)}, \cos \theta_\ell^{(0)}]^\top.
\end{align}
Throughout the work, we aim to precisely and quickly estimate every SP's triplet parameters, namely, $\{\tau_\ell, \theta^{(0)}_{\ell}, \phi^{(0)}_{\ell}\}_{\ell=1}^L$. {To this end}, we make the following assumption:

\begin{table}[t]
\centering
\caption{Experimental Validation of Assumption \ref{Assumption of AoA Prior}: AoA Statistics in Regions of Interest using Sionna Ray Tracing.}
\label{tab:aoa_prior_stats}
\setlength{\tabcolsep}{1.0pt}
\renewcommand{\arraystretch}{1.05}
\footnotesize                     
\scalebox{1}{%
\begin{tabular}{c|c|cc|cc|c}
\hline
\multirow{2}{*}{\shortstack{Region of\\Interest}} 
& \multirow{2}{*}{\shortstack{Signal\\Path}}
& $\mu_{\ell}$ & $\sigma_{\ell}$ & $\xi_{\ell}$ & $\varsigma_{\ell}$
& \multirow{2}{*}{\shortstack{Correlation\\Coefficient}} \\
\cline{3-6}
& & \multicolumn{2}{c|}{Elevation} & \multicolumn{2}{c|}{Azimuth} & \\
\hline
\multirow{3}{*}{\shortstack{Florence Duomo\\(Florence)}}
& SP 1 & $100.8^\circ$  & $4.6^\circ$ & $62.7^\circ$  & $10.4^\circ$ & -0.19 \\
& SP 2 & $95.1^\circ$  & $0.9^\circ$ & $93.2^\circ$  & $2.0^\circ$ & 0.14 \\
& SP 3 & $101.1^\circ$  & $4.5^\circ$ & $113.7^\circ$  & $4.5^\circ$ & 0.17 \\
\hline
\multirow{3}{*}{\shortstack{Inha Aerospace Campus \\(Incheon)}}
& SP 1 & $94.9^\circ$ & $1.1^\circ$ & $36.8^\circ$   & $9.9^\circ$ & -0.06 \\
& SP 2 & $93.7^\circ$ & $0.6^\circ$ & $84.2^\circ$  & $5.8^\circ$ & -0.17 \\
& SP 3 & $96.0^\circ$ & $1.7^\circ$ & $141.8^\circ$  & $13.6^\circ$ & 0.06 \\
\hline
\end{tabular}%
} {\vspace{-5mm}}
\end{table}

\begin{assumption}[Prior AoA Statistics]\label{Assumption of AoA Prior} \emph{The last reflection point $\boldsymbol{b}_{\ell}$ of SP $\ell$ typically lies on the surfaces of surrounding objects 
(e.g., walls and pillars), which are known in advance from a coarse environment model as exemplified in Table~\ref{tab:aoa_prior_stats}. Consequently, the induced AoA uncertainty can be approximated by a weakly informative Gaussian prior:
\begin{align}
\label{eq:aoa_statistics}
\theta_\ell^{(0)}\sim \mathcal{N}(\mu_{\ell}, \sigma_{\ell}^2), \quad \phi_\ell^{(0)}\sim\mathcal{N}(\xi_\ell,\varsigma_\ell^2), \quad \theta_\ell^{(0)}\bot \phi_\ell^{(0)},
\end{align}
where the corresponding mean and standard-deviation pairs, say
$\{\mu_{\ell}, \sigma_{\ell}\}$ and $\{\xi_\ell, \varsigma_{\ell}\}$, are
known a priori. The AoAs are further assumed to be mutually
independent across SPs.} 
\end{assumption} 
Note that Assumption \ref{Assumption of AoA Prior} does not restrict our estimator to these mean directions. It merely introduces soft statistical regularization. In other words, all AoA values remain admissible, and the prior only influences the relative likelihood of angles.

\begin{figure}[t]
    \centering
    \begin{subfigure}[b]{0.48\linewidth}
        \centering
        \includegraphics[height=3.71cm]{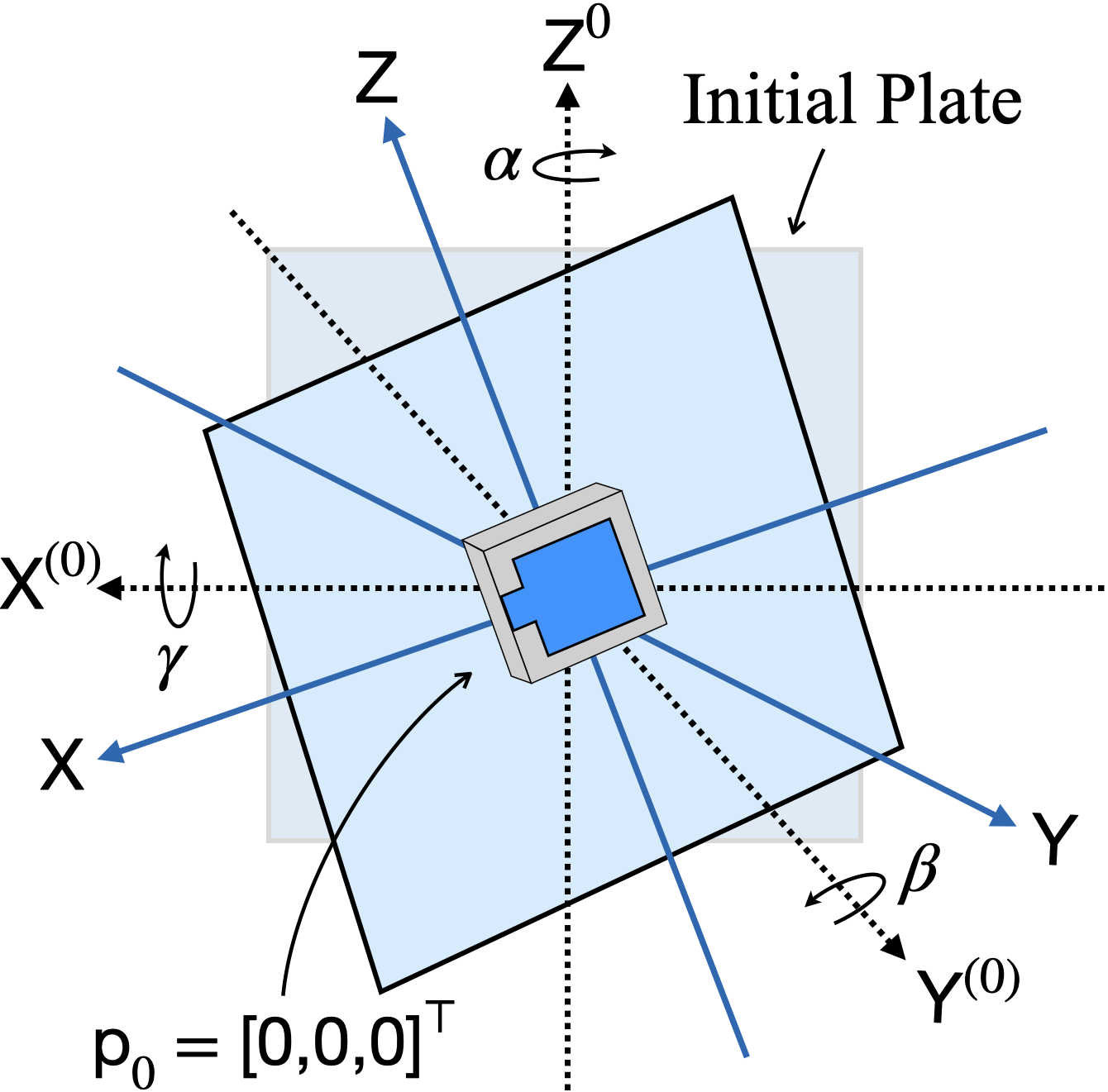}
        \caption{}
        \label{fig:coordinate_a}
    \end{subfigure}
    \hfill
    \begin{subfigure}[b]{0.48\linewidth}
        \centering
        \includegraphics[height=3.71cm]{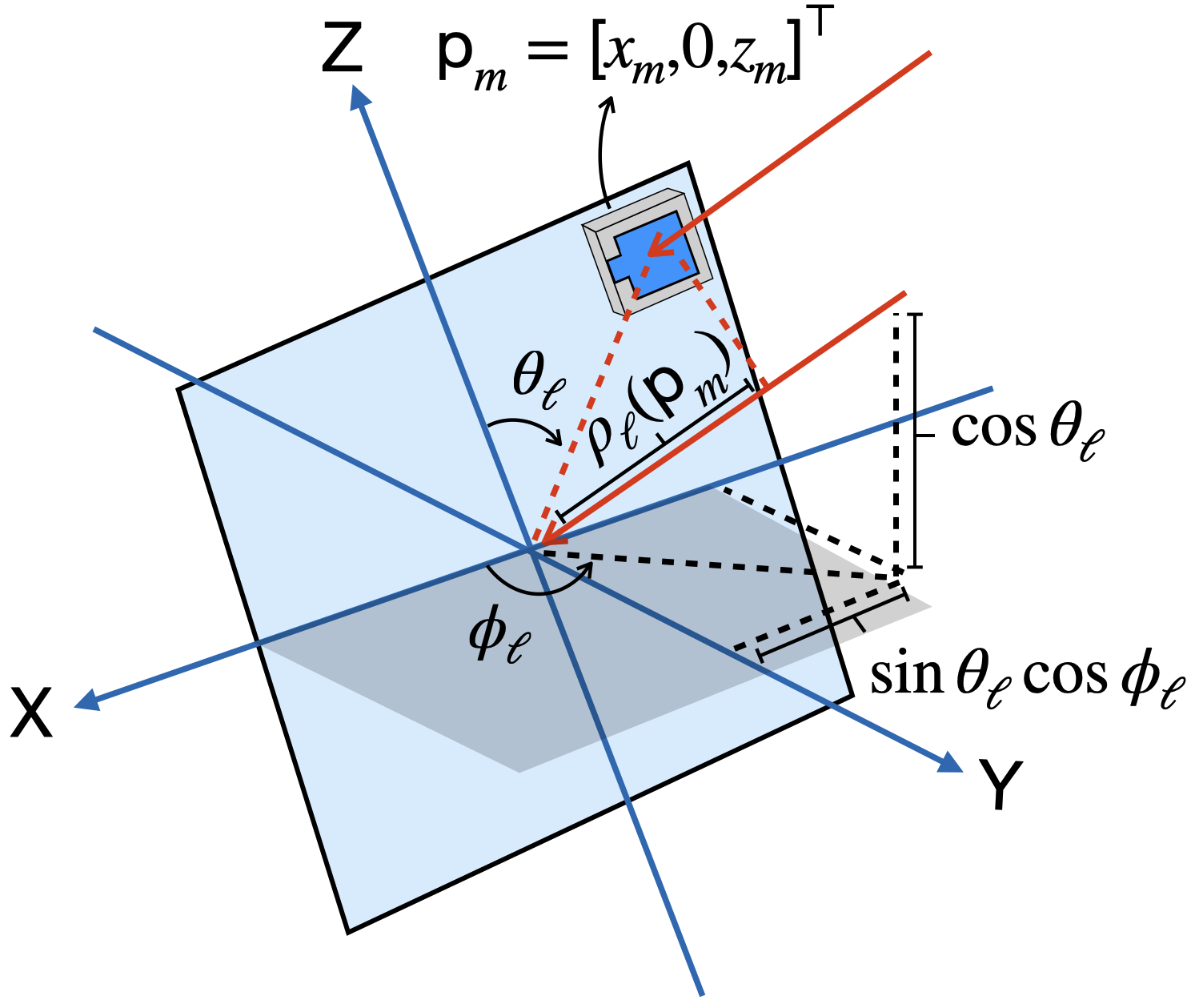}
        \caption{}
        \label{fig:coordinate_b}
    \end{subfigure}
    \caption{(a) The movable plate's orientation is configured by tilting angles
    ($\alpha, \beta, \gamma$) to define a rotated coordinate system.
    (b) The MA's movement to $\boldsymbol{p}_m$ on the movable plate induces
    an additional propagation distance
    $\rho_\ell(\boldsymbol{p}_m)$.\vspace{-15pt}}
    \label{fig:coordinate}
\end{figure}

\vspace{-10pt}
\subsection{Signal Model}\label{Sec: Signal Model}
\textit{1) Transmit Signal:}
We consider an \emph{orthogonal frequency-division multiplexing} (OFDM) system, where the transmitter sends the OFDM signals $\boldsymbol{x}(t) = [x^{(1)}(t), \dots, x^{(K)}(t)]^{\top} \in \mathbb{C}^{{K} \times {1}} $ through a set of ${K}$ subcarriers. With frequency spacing $\delta$ and frame duration $\mathrm{T} = {1/\delta}$, the $k$-th subcarrier's signal, denoted by $x^{(k)}(t)$, is given as
\begin{align}
x^{(k)}(t) = \sqrt{\frac{P}{K}}  e^{j 2 \pi f_k t}, \quad \ 0\le t < \mathrm{T},
\end{align}
where $P$ is the total transmit power and $f_k$ is the frequency of the $k$-th subcarrier, for $k=1, \dots, K$. These subcarriers are symmetrically centered around the carrier frequency $f_c$ with a spacing of $\delta$.

 \textit{2) Channel:}
The channel between the transmitter and the receiver is affected by (i) the movable plate's mechanical orientation and (ii) the MA's position within it. We describe these two factors as follows. 

The movable plate, initially aligned with the $\mathsf{X}^{(0)}$-$\mathsf{Z}^{(0)}$ plane, can be configured by applying tilting angles $\alpha$, $\beta$, and $\gamma$ about the $\mathsf{Z}^{(0)}$, $\mathsf{Y}^{(0)}$, and $\mathsf{X}^{(0)}$-axes, respectively. All tilting operations are defined with respect to the reference point $\boldsymbol{p}_0$, which serves as the rotation center. Once the tilting angles are specified, we define a rotated coordinate system ($\mathsf{X}$, $\mathsf{Y}$, $\mathsf{Z}$), in which the tilted plate lies on the $\mathsf{X}$-$\mathsf{Z}$ plane and its outward normal becomes $\mathsf{Y}$ (see Fig.~\ref{fig:coordinate}). From the perspective of the rotated plate, the arrival direction of SP $\ell$ in \eqref{eq: Initial AoA direction} appears as if the inverse rotation has been applied to the original direction vector, namely, 
\begin{align}\label{eq: revised AoA direction}
    {\boldsymbol{a}}_\ell =\boldsymbol{R}^{\top}(\alpha,\beta,\gamma)
    \boldsymbol{a}^{(0)}_\ell,
\end{align}
 where $\boldsymbol{R}(\alpha,\beta,\gamma)$ is the composite rotation matrix determined by the tilting angles \cite{Diebel2006}, given as
 \begin{align}
     \boldsymbol{R}&(\alpha,\beta,\gamma) = \boldsymbol{R}_z(\alpha) \boldsymbol{R}_y(\beta)\boldsymbol{R}_x(\gamma).
\end{align}
This matrix is orthogonal, and its inverse equals its transpose. Analogous to the definition of $\{\theta_\ell^{(0)}, \phi^{(0)}_\ell\}$ from $\boldsymbol{a}^{(0)}_{\ell}$ in \eqref{eq: Initial AoA direction}, the elevation and azimuth angles observed from the rotated plate, denoted by ${\theta}_\ell$ and ${\phi}_\ell$, are obtained from ${\boldsymbol{a}}_{\ell}$ in \eqref{eq: revised AoA direction}, given as
\begin{align}
    {\theta}_\ell
    = \cos^{-1}\!\left( \frac{[{\boldsymbol{a}}_\ell]_3}{\| \boldsymbol{a}_\ell\|}\right), 
    \quad
    {\phi}_\ell
    = {\tan^{-1}}\!\left(\frac{[{\boldsymbol{a}}_\ell]_2}{[{\boldsymbol{a}}_\ell]_1}\right). \label{eq:recover}
\end{align}   

Second, to model the effect of the MA's position, we consider that the MA moves from the reference point $\boldsymbol{p}_0$ to a new position $\boldsymbol{p}_m$. 
In the rotated coordinate system, this position is expressed as ${\boldsymbol{p}}_m=[{x}_m,0,{z}_m]^{\top}$, where the ${\mathsf{Y}}$-coordinate is zero because the MA is constrained to lie on the movable plate. We assume that the coherence time is longer than the MA movement duration, so that the channel parameters remain unchanged. We define the channel vector of SP $\ell$ at position $\boldsymbol{p}_m$ as $\boldsymbol{h}_\ell(\boldsymbol{p}_m) = [h_\ell^{(1)}(\boldsymbol{p}_m), \dots , h_\ell^{(K)}(\boldsymbol{p}_m) ] \in \mathbb{C}^{1 \times K}$. 
The element $h_{\ell}^{(k)}(\boldsymbol{p}_m)$ denotes the channel response of SP $\ell$ on the $k$-th subcarrier, given by  
\begin{align}
h_{\ell}^{(k)}(\boldsymbol{p}_m) = \alpha_\ell \, \exp \left( -j{2\pi}f_k\left(\tau_\ell+\frac{\rho_{\ell}(\boldsymbol{p}_m)}{c}\right) \right),
\end{align}
where $\alpha_\ell$ represents the attenuation factor reflecting path-loss and $c\approx 3\times 10^8$ (m/s) is the speed of light. The term $\rho_\ell(\boldsymbol{p}_m)$ represents the additional propagation distance of SP $\ell$ as the MA moves from $\boldsymbol{p}_0$ to $\boldsymbol{p}_m$, equivalent to the scalar projection of $\boldsymbol{p}_m$ onto the unit direction ${\boldsymbol{a}}_\ell$: 
\begin{align}\label{eq: propagation dist}
\rho_\ell(\boldsymbol{p}_m)
&= {\boldsymbol{a}}_\ell^\top {\boldsymbol{p}}_m 
= {x}_m \sin {\theta}_\ell \cos {\phi}_\ell
 + {z}_m \cos {\theta}_\ell.
\end{align}
Consequently, the overall channel at location $\boldsymbol{p}_m$ on the $k$-th subcarrier is
\begin{align}\label{eq: Overall channel}
h^{(k)}(\boldsymbol{p}_m)
= \sum_{\ell=1}^L \alpha_\ell
   \exp\!\left(
      -j 2\pi f_k \left(
         \tau_\ell + \frac{\rho_\ell(\boldsymbol{p}_m)}{c}
      \right)
   \right).
\end{align}
It is observed from \eqref{eq: Overall channel} that all estimation parameters are embedded in the phase terms. Their accurate estimation is viable when the phase contributions of different SPs are sufficiently distinct. The additional propagation distances $\{\rho_\ell\}$ in \eqref{eq: propagation dist} are controllable through the orientation of the movable plate as well as the MA's movement. This controllability enables us to enlarge inter-phase differences, as illustrated in the following example.     
\vspace{-5pt}
\begin{remark}[Effect of Plate Configuration]\label{Remark1} \emph{Consider an example with two SPs (\(L=2\)), whose elevation-azimuth AoAs are 
$(\theta_1^{(0)},\phi_1^{(0)})= (60^\circ,\,75^\circ)$ and
$(\theta_2^{(0)},\phi_2^{(0)})= (75^\circ,\,75^\circ)$. Suppose the MA moves a distance $d$ along the ${\mathsf{X}}$-axis, i.e., ${\boldsymbol{p}}_m=[d,0,0]^\top$. We compare the resulting differences in the additional propagation distance $\rho_\ell(\boldsymbol{p}_m)$ under two plate configurations: 
\begin{itemize}[leftmargin=*]
\item \textbf{Without plate tilting} ($\alpha=\beta=\gamma=0$): The difference between $\rho_1$ and $\rho_2$ in \eqref{eq: propagation dist} is
$|\rho_1(\boldsymbol{p}_m)-\rho_2(\boldsymbol{p}_m)|\approx 0.0259\times d
$.
\item \textbf{With plate tilting} configured as $(\alpha,\beta,\gamma)=(0^\circ,90^\circ,45^\circ)$: 
The above difference increases substantially to
$|\rho_1(\boldsymbol{p}_m)-\rho_2(\boldsymbol{p}_m)|\approx  0.2412\times d$.
\end{itemize}
This example shows the effect of the movable plate's orientation on multi-path resolvability. By appropriately configuring the tilting angles, the MA induces much larger differences in the projected path lengths, thereby significantly enhancing inter-path phase separation. As a result, even SPs with very similar AoAs and ToAs can be resolved more~effectively.}
\end{remark}

 \textit{3) Receive Signal:}
The received signal on the $k$-th subcarrier at time $t$, when the MA is at position $\boldsymbol{p}_m$, denoted by $y^{(k)}(t;\boldsymbol{p}_m)$, is given as 
\begin{align}
{y}^{(k)}\!(t;\boldsymbol{p}_m)\!=\!\! \sum_{\ell =1}^L \alpha_\ell 
   e^{\!-j2\pi f_k\left(
      \tau_\ell 
      + \frac{\rho_\ell(\boldsymbol{p}_m)}{c}
     \right)} 
   x^{(k)}\!(t) \!+\! w^{(k)}\!(t),
\end{align}
where $w^{(k)}(t) \in \mathbb{C}$ is the thermal noise on the $k$-th subcarrier, modeled as a circularly symmetric complex Gaussian white noise process, with independent samples following $\mathcal{CN}(0, N_0)$. Denote the demodulated received signal vector at position $\boldsymbol{p}_m$ as $\boldsymbol{s}(\boldsymbol{p}_m) = [s^{(1)}(\boldsymbol{p}_m), \dots ,s^{(K)}(\boldsymbol{p}_m)]^\top$, where the $k$-th element corresponds to the demodulated signal on the $k$-th subcarrier, obtained as
\begin{align}
{s}^{(k)}(\boldsymbol{p}_m) &= \frac{1}{\mathrm{T}} \int_{0}^{\mathrm{T}} y^{(k)}(t;\boldsymbol{p}_m) \cdot x^{(k)}(t)^* dt \nonumber \\
&= \frac{P}{K} \left ( \sum_{\ell =1}^L \alpha_\ell e^{-j2\pi f_k( \tau_\ell + {  \frac{\rho_\ell(\boldsymbol{p}_m)}{c} })} \right ) + \bar{w}^{(k)}, \label{eq:received_signal}
\end{align}
where $\bar{w}^{(k)}$ is the demodulated noise component following $\mathcal{CN}(0, PN_0/K)$.

\vspace{-10pt}
\subsection{Multi-Path Sensing: Two-Step Procedure}
\vspace{-2pt}
To estimate all SP parameters, namely the elevation-azimuth AoAs $\{\theta^{(0)}_{\ell}, \phi^{(0)}_{\ell}\}$ and the ToAs $\{\tau_\ell \}$, we will adopt the following two-step procedure:

\begin{enumerate}[leftmargin=*]
\item \textbf{Movable-Plate Orientation Control}: We first configure the movable plate's orientation $(\alpha, \beta,\gamma)$ to keep every SP visible while enhancing their separability, guided by the prior AoA statistics in Assumption \ref{Assumption of AoA Prior}. Once the orientation is determined, it remains fixed to avoid excessive tilting overhead \cite{Shao2025}. Under this configuration, the MA performs two linear scans along the $\mathsf{X}$- and $\mathsf{Z}$-axes of the rotated plate. The detailed explanation of this step is provided in Sec.~\ref{sec:Movable Plate Orientation Control}.

\item \textbf{SP Parameter Estimation}: Given the two linear scan measurements, we first obtain the estimates of the local AoAs, denoted by $\{\hat\theta_{\ell}, \hat\phi_{\ell}\}$. The corresponding global (initial-frame) AoAs $\{\hat\theta^{(0)}_{\ell}, \hat\phi^{(0)}_{\ell}\}$ are then recovered by applying the rotation transformation to the local estimates, according to the configured tilting angles. Finally, leveraging the recovered AoAs, we estimate each SP's ToA $\hat\tau_\ell$. The detailed explanation of this step is provided in Sec.~\ref{sec:SPsensing}.
\end{enumerate}

\section{Optimal Movable-Plate Orientation Control}\label{sec:Movable Plate Orientation Control}
This section focuses on controlling the movable plate's orientation to enable accurate estimation of each SP's AoA, characterized by its elevation and azimuth. To this end, we first analyze its effect on AoA estimation from a Fisher information perspective. Guided by the resulting insights, we formulate an optimal control problem for plate-orientation design and develop a tractable solution algorithm. 

\subsection{Fisher Information Analysis and Design Guidelines}\label{sec:FIM}
In this subsection, we derive the \emph{Fisher information matrix} (FIM) for the AoA estimation problem, which will serve as a guideline for the subsequent optimization problem. 

\textit{1) Fisher Information Derivation}: Let the unknown AoA parameter vector be formed by stacking the AoAs of all $L$ SPs as $\boldsymbol{\psi} = [\boldsymbol{\psi}_{1}^\top, \dots, \boldsymbol{\psi}_{L}^\top]^\top \in \mathbb{R}^{2L \times 1}$, where $\boldsymbol{\psi}_{\ell}=[\theta^{(0)}_{\ell}, \phi_{\ell}^{(0)}]^\top$ collects the elevation and azimuth angles of SP $\ell$ specified in \eqref{eq:angles}. Next, we model the MA as sequentially visiting $M$ arbitrary measurement positions $\{\boldsymbol{p}_m\}_{m=1}^M$ on the plate. 
The received signals at the $M$ positions, i.e., $\boldsymbol{s}(\boldsymbol{p}_1),\boldsymbol{s}(\boldsymbol{p}_2),\dots,\boldsymbol{s}(\boldsymbol{p}_M)$  in \eqref{eq:received_signal}, are stacked into the observation vector $\boldsymbol{s} \in \mathbb{C}^{MK \times 1}$ as
\begin{align}
    \boldsymbol{s} = [\boldsymbol{s}(\boldsymbol{p}_1)^\top, \boldsymbol{s}(\boldsymbol{p}_2)^\top, \dots, \boldsymbol{s}(\boldsymbol{p}_M)^\top]^\top.
\end{align}
Recalling the Gaussian model in \eqref{eq:received_signal}, the likelihood function $p(\boldsymbol{s}|\boldsymbol{\psi})$ is given by
\begin{align}
    p(\boldsymbol{s}|\boldsymbol{\psi}) = \frac{\exp\left( - \frac{K}{PN_0} \left\| \boldsymbol{s} - \boldsymbol{\mu}(\boldsymbol{\psi}) \right\|^2 \right)}{(\pi (P N_0/K))^{MK}} ,
\end{align}
where the mean signal vector $\boldsymbol{\mu}(\boldsymbol{\psi})$ is obtained by removing the noise term in \eqref{eq:received_signal}. The resultant FIM $\boldsymbol{I}(\boldsymbol{\psi})\in\mathbb{R}^{2L\times 2L}$ can be expressed as
\begin{align}\label{eq: Gamma}
    \boldsymbol{I}(\boldsymbol{\psi})
    = \frac{2K}{PN_0} \Re \left\{
    \begin{bmatrix}
    \boldsymbol{\Gamma}_{1,1} & \boldsymbol{\Gamma}_{1,2} & \dots & \boldsymbol{\Gamma}_{1,L} \\
    \boldsymbol{\Gamma}_{2,1} & \boldsymbol{\Gamma}_{2,2} & \dots & \boldsymbol{\Gamma}_{2,L} \\
    \vdots & \vdots & \ddots & \vdots \\
    \boldsymbol{\Gamma}_{L,1} & \boldsymbol{\Gamma}_{L,2} & \dots & \boldsymbol{\Gamma}_{L,L}
    \end{bmatrix}
    \right\},
\end{align}
where the $(\ell,u)$-th block $\boldsymbol{\Gamma}_{\ell,u} \in \mathbb{C}^{2 \times 2}$,  $\ell, u \in \{1, \dots, L\}$, is given by \begin{align}
\boldsymbol{\Gamma}_{\ell,u} &= \left(\frac{\partial \boldsymbol{\mu}(\boldsymbol{\psi})}{\partial \boldsymbol{\psi}_\ell}\right)^H \left(\frac{\partial \boldsymbol{\mu}(\boldsymbol{\psi})}{\partial \boldsymbol{\psi}_u}\right)\nonumber\\
&=
\begin{cases}
{c}_\ell \sum_{m=1}^{M} \boldsymbol{g}_{\ell}(\boldsymbol{p}_m)\boldsymbol{g}_{\ell}(\boldsymbol{p}_m)^{\top}, & \ell=u,\\
\sum_{m=1}^{M} {\kappa}_{\ell,u}(\boldsymbol{p}_m)\boldsymbol{g}_{\ell}(\boldsymbol{p}_m)\boldsymbol{g}_{u}(\boldsymbol{p}_m)^{\top}, & \ell\neq u,
\end{cases}
\end{align}
with the AoA-gradient vector $\boldsymbol{g}_{\ell}(\boldsymbol{p}_m)\triangleq
\left[\frac{\partial \rho_\ell(\boldsymbol{p}_m)}{\partial \theta_\ell^{(0)}}, \frac{\partial \rho_\ell(\boldsymbol{p}_m)}{\partial \phi_\ell^{(0)}}\right]^\top$ which represents the sensitivity of $\rho_\ell(\boldsymbol{p}_m)$ of \eqref{eq: propagation dist} with respect to $\boldsymbol{\psi}_\ell$. The diagonal-block scaling factor ${c}_\ell \triangleq |\alpha_\ell|^2 \sum_{k=1}^K \left( \frac{2\pi f_k}{c}\frac{P}{K} \right)^2$ is independent of the MA's position $\boldsymbol{p}_m$. On the other hand, the off-diagonal scaling factor ${\kappa}_{\ell,u}(\boldsymbol{p}_m)$ depends on $\boldsymbol{p}_m$ and is given by
\begin{align}\label{eq: kappa}
&\kappa_{\ell,u}(\boldsymbol{p}_m) 
\nonumber\\\triangleq& \alpha_\ell^* \alpha_u \left( \frac{P}{K} \right)^2\!\! \sum_{k=1}^K \!\!\left(\frac{2\pi f_k}{c} \right)^2 \!\!\! 
\exp\left(j{2\pi f_k \Delta_{\ell,u}(\boldsymbol{p}_m)}\right),
\end{align}
where 
\begin{align}\label{eq: Delta}
\Delta_{\ell,u}(\boldsymbol{p}_m)=\tau_\ell - \tau_u
+ \frac{\rho_\ell(\boldsymbol{p}_m) - \rho_u(\boldsymbol{p}_m)}{c},
\end{align}
which determines the phase difference between the two SPs.

\textit{2) Design Guideline}: It is well known that the estimation error covariance is lower-bounded by the inverse of the FIM, i.e., 
$\mathsf{cov}(\hat{\boldsymbol{\psi}})\succeq \boldsymbol{I}(\boldsymbol{\psi})^{-1}$. Taking determinants on both sides yields the volume bound, given as 
\begin{align}
\mathsf{det}(\mathsf{cov}(\hat{\boldsymbol{\psi}}))\geq \frac{1}{\mathsf{det}(\boldsymbol{I}(\boldsymbol{\psi}))},
\end{align}
which motivates the maximization of $\mathsf{det}(\boldsymbol{I}(\boldsymbol{\psi}))$. In view of the block structure of $\boldsymbol{I}(\boldsymbol\psi)$ specified in \eqref{eq: Gamma}, two strategies arise: (i) strengthening the diagonal blocks and/or (ii) suppressing the off-diagonal coupling blocks. We discuss these two directions as follows. \begin{itemize}[leftmargin=*]
\item \textbf{Effect of increasing the diagonal blocks.} The diagonal block $\boldsymbol\Gamma_{\ell,\ell}$ is determined by the scalar coefficient $c_\ell$ and the AoA-gradient vectors $\boldsymbol{g}_\ell(\boldsymbol{p}_m)$. Since $c_\ell$ is independent of the movable plate orientation, the only controllable factor is the magnitude of $\boldsymbol{g}_\ell(\boldsymbol{p}_m)$. However, enlarging $\boldsymbol{g}_\ell(\boldsymbol{p}_m)$ also amplifies the cross terms that appear in the off-diagonal block. It is thus concluded that simply increasing the AoA-gradient magnitudes does not necessarily increase $\mathsf{det}(\boldsymbol{I}(\boldsymbol{\psi}))$ because the gain in $\boldsymbol\Gamma_{\ell,\ell}$ may be offset by strengthened inter-parameter coupling. 
\item \textbf{Effect of decreasing the off-diagonal blocks.} For $\ell\neq u$, the off-diagonal block $\boldsymbol\Gamma_{\ell,u}$ is proportional to $\kappa_{\ell,u}(\boldsymbol{p}_m)$ in \eqref{eq: kappa}. Under the narrowband approximation, it becomes 
\begin{align}
\kappa_{\ell,u}(\boldsymbol{p}_m)&\approx \alpha_\ell^* \alpha_u \left( \frac{P}{K} \right)^2 \left( \frac{2\pi f_c}{c} \right)^2 \times \nonumber\\
&\sum_{k=1}^K \exp\left( j {2\pi f_k} \Delta_{\ell,u}(\boldsymbol{p}_m)  \right).
\end{align}
The harmonic sum admits a closed form of the Dirichlet kernel, given as
\begin{align}
&\sum_{k=1}^K \exp\left( j {2\pi f_k} \Delta_{\ell,u}(\boldsymbol{p}_m)  \right)\nonumber\\
=&\exp(j 2\pi f_c \Delta_{\ell,u}(\boldsymbol{p}_m))D_K(\Delta_{\ell,u}(\boldsymbol{p}_m)), \label{eq:dirichlet}
\end{align}
where $D_K(x)\triangleq\frac{\sin \left(K \pi \delta x \right)}{\sin \left(\pi \delta x \right)}$ and the subcarrier spacing $\delta$ is specified in Sec.~\ref{Sec: Signal Model}. Since $|D_K(x)|$ decreases as $|x|$ moves away from zero (outside its main lobe), the magnitude of $\kappa_{\ell,u}(\boldsymbol{p}_m)$ and thus $\boldsymbol{\Gamma}_{\ell,u}$ are reduced when $|\Delta_{\ell,u}(\boldsymbol{p}_m)|$ becomes larger. As indicated in \eqref{eq: Delta}, $|\Delta_{\ell,u}(\boldsymbol{p}_m)|$ increases with $|\rho_\ell(\boldsymbol{p}_m) - \rho_u(\boldsymbol{p}_m)|$, which is consistent with the intuition in Remark \ref{Remark1}: enlarging the path-length separation suppresses the inter-path coupling even though their propagation delays are similar and improves $\mathsf{det}(\boldsymbol{I}(\boldsymbol{\psi}))$. 
\end{itemize}

\subsection{Stochastic Optimization of Movable-Plate Control} \label{sec:opt}
As noted in Remark \ref{Remark1}, the optimal plate orientation depends on the instantaneous AoA realization of each SP, which is generally unavailable in practice. We therefore adopt two linear MA scans along the $\mathsf{X}$- and $\mathsf{Z}$-axes with uniform step size $d$, namely, 
\begin{align}\label{eq:two_scans}  
\boldsymbol{p}_m^{(1)} \!=\! (m-1)d\mathbf{i}, \ 
\boldsymbol{p}_m^{(2)} \!=\! (m-1)d\mathbf{k}, \  m\!=\!1,\dots,M,
\end{align}
where $\mathbf{i}=[1,0,0]^\top$ and $\mathbf{k}=[0,0,1]^\top$ are standard basis vectors along the $\mathsf{X}$- and $\mathsf{Z}$-axes, respectively. Collectively, these two scan lines traverse the plate in orthogonal directions, providing broad spatial coverage with respect to the AoA.

Given the plate orientation $\boldsymbol\varphi=[\alpha,\beta,\gamma]$, the propagation distance variations of SP $\ell$ induced by the MA movements $\boldsymbol{p}_m^{(1)}$ and $\boldsymbol{p}_m^{(2)}$ are 
\begin{align}\label{eq:scan}
\rho_{\ell}(\boldsymbol{p}_m^{(1)})&=(m-1) d  \rho_{\ell}(\mathbf{i})=(m-1) d  [\boldsymbol{R}^{\top}(\boldsymbol{\varphi})
    \boldsymbol{a}^{(0)}_\ell]_1  ,\nonumber\\
\rho_{\ell}(\boldsymbol{p}_m^{(2)})&=(m-1) d \rho_{\ell}(\mathbf{k})=(m-1) d  [\boldsymbol{R}^{\top}(\boldsymbol{\varphi})
    \boldsymbol{a}^{(0)}_\ell]_3.
\end{align}
Both quantities scale linearly with the travel distance $(m-1)d$. Accordingly, $d$ is selected up to its permissible limit dictated by mechanical constraints and the physical size of the plate, while keeping the spacing within half the carrier wavelength to avoid spatial aliasing.

On the other hand, the remaining terms $\rho_{\ell}(\mathbf{i})$ and $\rho_{\ell}(\mathbf{k})$ are the scalar projections of the rotated arrival vector $\boldsymbol{R}^{\top}(\boldsymbol{\varphi})\boldsymbol{a}^{(0)}_\ell$ onto the unit directions $\mathbf{i}$ ($\mathsf{X}$-axis) and $\mathbf{k}$ ($\mathsf{Z}$-axis), respectively. 
Equivalently, 
\begin{align}
\rho_{\ell}(\mathbf{i})=\mathbf{i}^\top\boldsymbol{R}^{\top}(\boldsymbol{\varphi})\boldsymbol{a}^{(0)}_\ell,\quad
\rho_{\ell}(\mathbf{k})=\mathbf{k}^\top\boldsymbol{R}^{\top}(\boldsymbol{\varphi})\boldsymbol{a}^{(0)}_\ell.
\end{align}
Since $\boldsymbol{\varphi}$ affects the propagation-distance variation through these two projections, optimizing $\boldsymbol{\varphi}$ amounts to shaping $\rho_{\ell}(\mathbf{i})$ and $\rho_{\ell}(\mathbf{k})$.
We therefore focus on \(\rho_{\ell}(\mathbf{i})\) and \(\rho_{\ell}(\mathbf{k})\) as the key quantities governed by $\boldsymbol{\varphi}$. For later use, we also define the projection onto the unit direction $\mathbf{j}$ ($\mathsf{Y}$-axis), obtained by replacing $\mathbf{i}$ or $\mathbf{k}$ with $\mathbf{j}=[0,1,0]^\top$.

Since \(\boldsymbol{a}_{\ell}^{(0)}\) is random under the prior in Assumption~\ref{Assumption of AoA Prior}, the projections \(\rho_\ell(\mathbf{i})\), \(\rho_\ell(\mathbf{k})\), and \(\rho_\ell(\mathbf{j})\) are also random. To characterize this randomness, we derive their first and second moments in the following lemma.  

\begin{lemma}[First and Second Moments of the Projection]\label{lemma1} \emph{Given the Gaussian prior of elevation and azimuth angles stated in Assumption~\ref{Assumption of AoA Prior}, namely $\theta_\ell^{(0)}\sim\mathcal{N}(\mu_{\ell},\sigma_{\ell}^2)$ and
$\phi_\ell^{(0)}\sim\mathcal{N}(\xi_\ell,\varsigma_\ell^2)$, define
\begin{align}\label{eq:moment_def}
\bar\rho_{\ell}^{(1)} &= \mathsf{E}[\rho_{\ell}(\mathbf{i})], \quad
\bar\rho_{\ell}^{(2)} = \mathsf{E}[\rho_{\ell}(\mathbf{k})], \quad
\bar\rho_{\ell}^{(3)} = \mathsf{E}[\rho_{\ell}(\mathbf{j})], \nonumber\\
\nu_{\ell}^{(1)} &= \mathsf{E}[\rho_{\ell}(\mathbf{i})^2], \ \
\nu_{\ell}^{(2)} = \mathsf{E}[\rho_{\ell}(\mathbf{k})^2], \
\nu_{\ell}^{(3)} = \mathsf{E}[\rho_{\ell}(\mathbf{j})^2].
\end{align}
Because \(\theta_\ell^{(0)}\) and \(\phi_\ell^{(0)}\) are independent Gaussian random variables, these moments admit closed-form expressions via their characteristic functions, with the details omitted.}
\end{lemma}

With the closed-form moments of Lemma~\ref{lemma1}, we design the objective and constraints of the optimization problem introduced in the sequel.
\subsubsection{Objective Function}
The projections are random under the prior, so the separation between SPs should hold reliably across realizations rather than only on average. Accordingly, different SPs should have well-separated mean projections, while the dispersion of each projection remains small. Since several $\boldsymbol\varphi$-dependent terms such as $\sin\beta\cos\gamma\cos(\alpha-\xi_\ell)$ appear in both the first and second moments, the mean separation and dispersion cannot be controlled independently. As a result, even when the mean projections are well separated, the dispersion of some SP pairs can become comparable to, or even larger than, their mean gap, so that their projections overlap in individual realizations. In such cases, the realized ordering of two SP projections may become opposite to their mean ordering, which we refer to as an \emph{order-reversal} event.

\begin{figure}[t]
  \centering
  \includegraphics[width=0.9\linewidth]{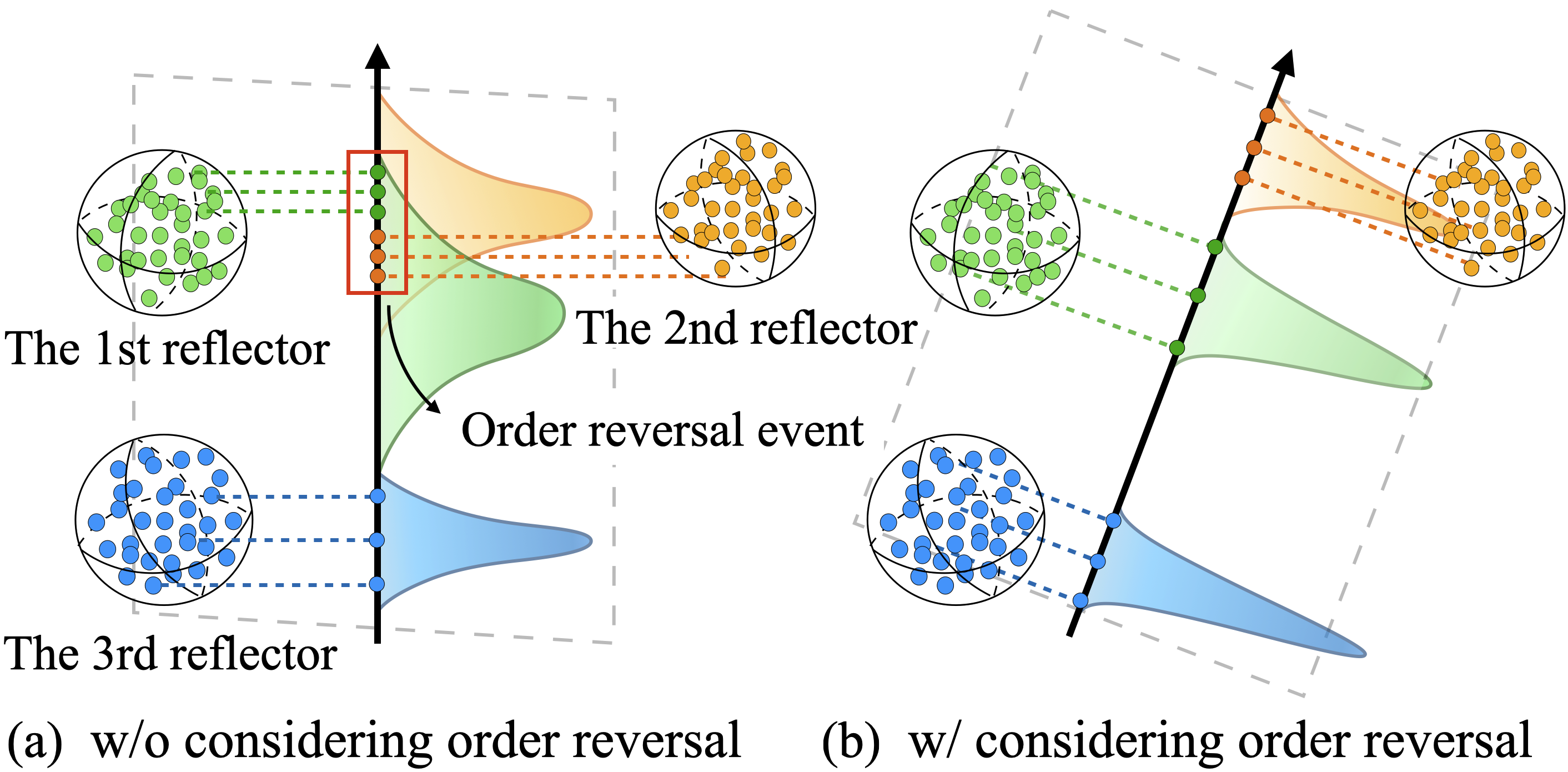}
  \caption{The effect of MA control on order-reversal.}\vspace{-15pt}
  \label{fig:epsilon_separation}
\end{figure}

\begin{definition}[Order-Reversal Probability]\label{Definition1}
\emph{Consider two SPs $\ell<u$. The probability of order-reversal on
the $\mathsf{X}$-axis is
\begin{align}
A_{\ell,u}^{(1)}\triangleq\mathsf{Pr}\!\left[\big(\rho_{\ell}(\mathbf{i})-\rho_{u}(\mathbf{i})\big)
\big(\bar\rho_{\ell}^{(1)}-\bar\rho_{u}^{(1)}\big)\le 0\right].
\end{align}
Similarly, the order-reversal probability on the $\mathsf{Z}$-axis is
\begin{align}
A_{\ell,u}^{(2)}\triangleq\mathsf{Pr}\!\left[\big(\rho_{\ell}(\mathbf{k})-\rho_{u}(\mathbf{k})\big)
\big(\bar\rho_{\ell}^{(2)}-\bar\rho_{u}^{(2)}\big)\le 0\right].
\end{align}}
\end{definition}
Fig.~\ref{fig:epsilon_separation} graphically illustrates a case where two SPs exhibit a large order-reversal probability, even though their mean projections are well separated. This occurs when the associated dispersions are large, so that the two projections frequently cross each other and thereby degrade the inter-path discriminability. To keep the SPs distinguishable, the order-reversal probability must therefore be controlled. Prompted by the definition, we minimize the aggregate log order-reversal
probability over all SP pairs and both scan axes, given by
\begin{align}
\sum_{\ell<u}\sum_{n=1}^2\log A_{\ell,u}^{(n)},
\end{align}
where the logarithm balances the discriminability across all pairs in a
proportional-fairness manner~\cite{Kelly1998ProportionalFairness}. Directly
optimizing this objective, however, requires the order-reversal probabilities
$A_{\ell,u}^{(n)}$, whose stochastic nature in Definition~\ref{Definition1} makes
an exact closed-form characterization intractable. We thus adopt the following
closed-form upper bound as a relaxation.

\begin{proposition}[Upper Bound on the Order-Reversal Probability]\label{prop:cantelli}
\emph{Consider two SPs $\ell$ and $u$. If $\bar\rho_{\ell}^{(1)}\geq\bar\rho_{u}^{(1)}$,
applying the Cantelli inequality~\cite{Boucheron2013Concentration} upper-bounds the order-reversal probability on the $\mathsf{X}$-axis as
\begin{align}\label{eq:cantelli}
A_{\ell, u}^{(1)}
\leq
\frac{\nu_{\ell}^{(1)}+\nu_{u}^{(1)}-(\bar\rho_{\ell}^{(1)})^2-(\bar\rho_{u}^{(1)})^2}
{\nu_{\ell}^{(1)}+\nu_{u}^{(1)}-2\bar\rho_{\ell}^{(1)}\bar\rho_{u}^{(1)}},
\end{align}
where all components are specified in Lemma~\ref{lemma1}. The same bound holds when
$\bar\rho_{u}^{(1)}>\bar\rho_{\ell}^{(1)}$, and the upper bound of $A_{\ell,u}^{(2)}$
is obtained by replacing the superscript $(1)$ with $(2)$.}
\end{proposition}
 \begin{proof}
Define the difference random variable
$V_{\ell,u}^{(1)}=\rho_{\ell}(\mathbf{i})-\rho_{u}(\mathbf{i})$, whose mean and variance are
\begin{align}
\mathsf{E}[V_{\ell,u}^{(1)}]&=\bar\rho_{\ell}^{(1)}-\bar\rho_{u}^{(1)},\label{Proof_Proposition1_1}\\
\mathsf{var}[V_{\ell,u}^{(1)}]&=\mathsf{var}[\rho_{\ell}(\mathbf{i})]+\mathsf{var}[\rho_{u}(\mathbf{i})]\nonumber\\
&=\nu_{\ell}^{(1)}-(\bar\rho_{\ell}^{(1)})^2+\nu_{u}^{(1)}-(\bar\rho_{u}^{(1)})^2.\label{Proof_Proposition1_2}
\end{align}
Since the order-reversal event is equivalent to $V_{\ell,u}^{(1)}\le 0$, the Cantelli
inequality gives
\begin{align}
\mathsf{Pr}[V_{\ell,u}^{(1)}\le 0]\leq
\frac{\mathsf{var}[V_{\ell,u}^{(1)}]}{\mathsf{var}[V_{\ell,u}^{(1)}]+(\mathsf{E}[V_{\ell,u}^{(1)}])^2}.
\end{align}
Substituting \eqref{Proof_Proposition1_1} and \eqref{Proof_Proposition1_2} yields~\eqref{eq:cantelli}. The case $\bar\rho_{u}^{(1)}>\bar\rho_{\ell}^{(1)}$ follows identically by interchanging $\ell$ and $u$. The bound for $A_{\ell,u}^{(2)}$ follows analogously by using $V_{\ell,u}^{(2)}=\rho_{\ell}(\mathbf{k})-\rho_{u}(\mathbf{k})$,
which completes the proof.
 \end{proof}
Summing the logarithm of the Cantelli bound in~\eqref{eq:cantelli} over all SP
pairs on both scan axes gives
\begin{align}\label{eq:logsum}
\!\!\!\!\sum_{\ell<u}\!\sum_{n=1}^{2} \log\! A_{\ell,u}^{(n)} \!\leq\!
\sum_{\ell<u}\!\sum_{n=1}^{2}
\log\!\frac{\nu_{\ell}^{(n)}\!\!+\!\nu_{u}^{(n)}\!\!\!-\!\!(\bar\rho_{\ell}^{(n)})^2\!\!-\!\!(\bar\rho_{u}^{(n)})^2}
{\nu_{\ell}^{(n)}\!\!+\!\nu_{u}^{(n)}\!\!-\!2\bar\rho_{\ell}^{(n)}\bar\rho_{u}^{(n)}}\!.
\end{align}
Reducing this upper bound suppresses the order-reversal probabilities of all SP
pairs. Since minimizing it is equivalent to maximizing its negation, we recast
\eqref{eq:logsum} as the maximization objective
\begin{align}\label{eq:objective_rewrite}
f(\boldsymbol\varphi)
\!=\! \sum_{\ell < u}\sum_{n=1}^{2}
\log\!\left(\!1\!+\!\frac{(\bar\rho_{\ell}^{(n)}\!\!-\!\bar\rho_{u}^{(n)})^2}{\nu_{\ell}^{(n)}\!\!-\!(\bar\rho_{\ell}^{(n)})^2\!+\!\nu_{u}^{(n)}\!\!-\!(\bar\rho_{u}^{(n)})^2}\!\right)\!.
\end{align}
Here, the numerator is the squared mean separation between the two SP projections, and the denominator is their total dispersion. Maximizing $f(\boldsymbol\varphi)$ thus enlarges the inter-SP separation while reducing the dispersion on both scan axes, thereby minimizing the upper bound on the order-reversal probability of every
SP pair.

\subsubsection{Front-Side Constraint}
Since the MA is mounted on only one side of the movable plate, each SP should arrive from the front side of the plate to ensure reliable reception. 
SP $\ell$ is thus said to satisfy the \emph{front-side condition} if
\begin{align}
\mathsf{Pr}\!\left[\rho_\ell(\mathbf{j})\le 0\right]\le \epsilon,
\label{eq:front_prob}
\end{align}
where $\epsilon\,(\ll 1)$ denotes the maximum allowable probability of violating the front-side condition.
A tractable sufficient condition for this constraint  can be derived below. 

\begin{proposition}[Sufficient Condition for Front-Side Incidence]\label{prop:front-side-condition}
\emph{A sufficient condition for all SPs to satisfy the front-side incidence constraint is }
\begin{align}\label{eq:front-side sufficient condition}\tag{C1}
c_{\ell}(\boldsymbol{\varphi})\ge 0,
\end{align}
\emph{where}
$
c_{\ell}(\boldsymbol{\varphi})
=
\bar\rho_{\ell}^{(3)}
-
\sqrt{
\frac{1-\epsilon}{\epsilon}
\left(
\nu_{\ell}^{(3)}-\big(\bar\rho_{\ell}^{(3)}\big)^2
\right)
}$.
\emph{Here, the notations are defined in Lemma~\ref{lemma1}.}
\end{proposition}
 \begin{proof} 
Using the Cantelli bound, $\mathsf{Pr}\!\left[\rho_\ell(\mathbf{j})\le 0\right]$ is written as 
\begin{align}
\mathsf{Pr}[\rho_{\ell}(\mathbf{j})\le 0]
{\leq}
\frac{\nu_\ell^{(3)}-(\bar{\rho}^{(3)}_\ell)^2 }
{\nu_\ell^{(3)}}. \nonumber
\end{align}
The sufficient condition of \eqref{eq:front_prob} is thus given as 
$\frac{\nu_\ell^{(3)}-(\bar{\rho}^{(3)}_\ell)^2 }
{\nu_\ell^{(3)}}
\le \epsilon$, which can be converted to $c_{\ell}(\boldsymbol{\varphi}) \ge 0 $.
 \end{proof}

Consequently, we formulate the plate-orientation problem~as
\begin{align}
&\max_{\boldsymbol\varphi=[\alpha,\beta,\gamma]} f(\boldsymbol{\varphi})\nonumber\\
\text{s.t. } & c_{\ell}(\boldsymbol{\varphi})\geq0, \quad \forall \ell,\ 
\label{eq: problem formualtion 1}\tag{P1}
\end{align}
whose solution will be derived in the following subsection. 

\subsection{Sequential Quadratic Programming Approach}\label{sec:Proposed Algorithm1}
Since \ref{eq: problem formualtion 1} involves a non-convex objective and a non-convex inequality constraint, we adopt SQP, which addresses the non-convexity by solving a sequence of local \emph{quadratic programs} (QPs) constructed through approximations of the objective and the constraint. 

Under standard regularity conditions, the iterates generated by SQP converge to a point satisfying the first-order stationary condition
\begin{align}
\nabla_{\boldsymbol{\varphi}}\mathcal{L}(\boldsymbol{\varphi}^\star, \boldsymbol{\eta}^\star)=\boldsymbol{0},
\end{align}
where the Lagrangian is defined as
\begin{align} \label{eq:lagrangian}
    \mathcal{L}(\boldsymbol{\varphi}, \boldsymbol{\eta}) 
    = f(\boldsymbol{\varphi}) + \sum_{\ell=1}^{L} \eta_\ell\, c_\ell(\boldsymbol{\varphi}),
\end{align}
with $\boldsymbol{\eta}$ denoting the vector of nonnegative Lagrange multipliers with elements $\eta_\ell$.

Let us denote by $\boldsymbol{\varphi}_t$ the current solution at the $t$-th iteration, which is updated in the direction of $\boldsymbol{d}_t$~as 
\begin{align}\label{eq: variable update}
\boldsymbol{\varphi}_{t+1} = \boldsymbol{\varphi}_t + \kappa_t \boldsymbol{d}_t,
\end{align}
where $\kappa_t \in (0, 1]$ is the step size chosen by a globalization strategy (e.g., a line search based on a merit function). Then, the objective function is approximated around $\boldsymbol{\varphi}_t$ by the second-order Taylor expansion as
\begin{align}\label{eq: Hessian based expansion}
f(\boldsymbol{\varphi}_t+\boldsymbol{d}_t)
\approx f(\boldsymbol{\varphi}_t)+ \nabla f(\boldsymbol{\varphi}_t)^\top \boldsymbol{d}_t+\frac{1}{2}\boldsymbol{d}_t^\top\boldsymbol{B}_t\boldsymbol{d}_t,
\end{align}
where $\boldsymbol{B}_t$ is the approximated Hessian of the Lagrangian with respect to $\boldsymbol{\varphi}$, i.e.,  
$\boldsymbol{B}_t \approx \nabla_{\boldsymbol{\varphi}\boldsymbol{\varphi}}^2 \mathcal{L}(\boldsymbol{\varphi}_t, \boldsymbol{\eta}_t)$, which is updated using the well-known \emph{Broyden–Fletcher–Goldfarb–Shanno} (BFGS) algorithm to preserve symmetry and maintain negative definiteness \cite{Nocedal2006}.

Next, the constraint in \ref{eq: problem formualtion 1} is approximated by the first-order Taylor expansion as $c_\ell(\boldsymbol{\varphi}_t+\boldsymbol{d}_t)\approx c_\ell(\boldsymbol{\varphi}_t)+\nabla c_\ell(\boldsymbol{\varphi}_t)^\top \boldsymbol{d}_t$, which becomes the following linear constraint:
\begin{align}\label{eq: gradient based expansion}
c_\ell(\boldsymbol{\varphi}_t)+\nabla c_\ell(\boldsymbol{\varphi}_t)^\top \boldsymbol{d}_t\geq 0,\quad  \forall \ell.
\end{align}

Combining \eqref{eq: Hessian based expansion} and \eqref{eq: gradient based expansion} leads to the following local QP problem for the $t$-th iteration:
\begin{align} \label{eq: problem formualtion 2}
    &\max_{\boldsymbol{d}_t}  \quad \nabla f(\boldsymbol{\varphi}_t)^\top \boldsymbol{d}_t +\frac{1}{2} \boldsymbol{d}_t^\top \boldsymbol{B}_t \boldsymbol{d}_t, \quad
    \text{s.t. } \eqref{eq: gradient based expansion}.  \tag{P2}
\end{align}
The resulting QP subproblem \ref{eq: problem formualtion 2} can be efficiently solved by standard linearly constrained QP solvers (e.g., an active-set method), which yields the primal solution $\boldsymbol{d}_t^\star$ and the associated dual variables $\boldsymbol{\eta}_t^\star$. The direction $\boldsymbol{d}_t^\star$ is used to update $\boldsymbol{\varphi}_{t+1}$ according to \eqref{eq: variable update}, while $\boldsymbol{\eta}_t^\star$ is used to refine the curvature approximation $\boldsymbol{B}_{t+1}$. With an appropriate step-size control (e.g., a merit-function line search), SQP is guaranteed to converge to a first-order stationary point of \ref{eq: problem formualtion 1}, denoted by $\boldsymbol{\varphi}^\star\triangleq(\alpha^\star,\beta^\star, \gamma^\star)$. 

\section{AoA and ToA Estimation under Two Linear Movable Antenna Scans}\label{sec:SPsensing}
Based on the optimized orientation $\boldsymbol{\varphi}^\star$, this section estimates
the SP parameters $\{(\tau_\ell,\theta^{(0)}_{\ell},\phi^{(0)}_{\ell})\}_{\ell=1}^L$ from the two
linear scans specified in~\eqref{eq:two_scans}, along which the SP projections are well separated.
We first estimate the AoA-related parameters on the \(\mathsf{X}\)- and \(\mathsf{Z}\)-axis scans, which are matched across the axes through both the received signals and the prior AoA statistics to recover each SP's elevation and azimuth AoAs. Finally, we estimate each ToA via a spatial filter that amplifies the target SP while suppressing the others.

\subsection{Spatial-Frequency Parameter Extraction}\label{sec:EstimateAOA}
With the movable plate oriented at the optimized angles $\boldsymbol{\varphi}^\star$,
we estimate two types of AoA-related parameters, referred to as \emph{spatial-frequency parameters} (SFPs). Specifically, when the MA is located at $\boldsymbol{p}_m$, the additional propagation distance in \eqref{eq: propagation dist} is rewritten as 
\begin{align}
\rho_\ell(\boldsymbol{p}_m)
= x_m  {u}_\ell + z_m {v}_\ell,
\label{eq:rho_uv_def}
\end{align}
where \({u}_\ell\triangleq \sin\theta_\ell\cos\phi_\ell\) and
\({v}_\ell\triangleq \cos\theta_\ell\) are the SFPs.  Recall that 
the MA collects \(M\) samples along each of the \(\mathsf X\)- and \(\mathsf Z\)-axis scans. The parameters \({u}_\ell\) and \({v}_\ell\) are then estimated separately from the measurements obtained along the \(\mathsf X\)- and \(\mathsf Z\)-axes using the MUSIC algorithm \cite{Schmidt1986}, while omitting the details. This yields the two SFP sets
\begin{align}
\mathcal{U} = \{ \hat{{u}}^{(1)}, \dots,  \hat{{u}}^{(L)}\}, \quad
\mathcal{V} = \{ \hat{{v}}^{(1)}, \dots,  \hat{{v}}^{(L)}\},
\end{align}
where \(L\) denotes the number of SPs\footnote{
Recovering the AoA of an SP requires pairing one SFP from $\mathcal{U}$ with one from $\mathcal{V}$. Accordingly, the number of resolvable SPs is determined by the smaller number of SFPs detected along the two axes, which is assumed to be equal to $L$ for ease of explanation.}. The indices of the detected SFPs in \(\mathcal{U}\) and \(\mathcal{V}\) simply
label the extracted components and do not correspond to the SP indices.
Therefore, they should be properly paired to recover the elevation and azimuth
AoAs, as discussed in the following subsection.

\subsection{AoA Estimation via Spatial-Frequency Parameter Pairing}\label{sec:SFP_pairing}
Since the estimated SFPs in \(\mathcal U\) and \(\mathcal V\) are unordered, 
the $L$ SFPs in each set admit \(L!\) possible orderings. Let \(\mathcal S\) denote the set of all permutations of $\{1,\dots, L\}$. We aim to identify the permutation pair \((\Theta^\star,\Xi^\star)\in\mathcal S\times\mathcal S\) that
aligns the SFPs in \(\mathcal U\) and \(\mathcal V\) with the corresponding SP indices. Under this pairing, the SFP pair associated with SP $\ell$ is
$(\hat{{u}}^{(\Theta^\star_\ell)},\hat{{v}}^{(\Xi^\star_\ell)})$ and the complete set of paired SFPs is thus given by 
\begin{align}
\{(\hat{{u}}^{(\Theta^\star_\ell)},\hat{{v}}^{(\Xi^\star_\ell)})\}_{\ell=1}^L.
\end{align}

For correct permutation pairing, we leverage two clues: the first is the sequence of received signals obtained along two linear MA scans in \eqref{eq:scan}, while the second is the prior AoA statistics in Assumption~\ref{Assumption of AoA Prior}, whose details are given as follows.

\subsubsection{Received-Signal Log-Likelihood}
For a given permutation pair \((\Theta,\Xi)\) and the resultant SFP pairs \(\{(\hat{{u}}^{(\Theta_\ell)},\hat{{v}}^{(\Xi_\ell)})\}_{\ell=1}^L\), we construct the steering matrix to model the received signal across the \(2M\) positions as
\begin{align}
\boldsymbol\Upsilon^{(k)}
=
\big[\boldsymbol\Phi^{(k)}_1(\Theta,\Xi),\dots,\boldsymbol\Phi^{(k)}_L(\Theta,\Xi)\big],
\end{align}
where \(\boldsymbol\Phi_\ell^{(k)}(\Theta,\Xi)\in\mathbb C^{2M\times 1}\) denotes the steering vector of candidate SP $\ell$. 
The \(m\)-th element, corresponding to the phase response of SP $\ell$ at \(\boldsymbol p_m\), is given by
\begin{align}
\big[\boldsymbol\Phi_\ell^{(k)}(\Theta,\Xi)\big]_m
\!=\!
\exp\!\left(
-j2\pi f_k
\frac{x_m \hat{{u}}^{(\Theta_\ell)} + z_m \hat{{v}}^{(\Xi_\ell)}}{c}
\right),
\end{align}
where $x_m$ and $z_m$ are specified in \eqref{eq:rho_uv_def}.

We stack the received signals on the \(k\)-th subcarrier into
\begin{align}
\boldsymbol{s}^{(k)} = [s^{(k)}(\boldsymbol p_1^{(1)}),\dots,s^{(k)}(\boldsymbol p_M^{(1)}),s^{(k)}(\boldsymbol p_1^{(2)}),\dots,s^{(k)}(\boldsymbol p_M^{(2)})]^\top. \nonumber
\end{align}
Under the Gaussian model in \eqref{eq:received_signal}, the log-likelihood of the received signal is
\begin{align}
\mathcal L(\Theta,\Xi)
&\triangleq
\log {p}\!\left(\{\boldsymbol s^{(k)}\}_{k=1}^{K}\mid \Theta,\Xi,\{{\hat{\boldsymbol{q}}}^{(k)}\}_{k=1}^{K}\right)
\nonumber\\
&=
-\frac{\sum_{k=1}^{K}
\left\|
\boldsymbol s^{(k)}-\boldsymbol\Upsilon^{(k)} \hat{\boldsymbol{q}}^{(k)}
\right\|_2^2}{PN_0/K}
+\mathcal{C}_1,
\label{eq:L_data}
\end{align}
where \(\mathcal{C}_1=-2KM\log\!\left(\pi PN_0/K\right)\). The vector \(\boldsymbol {\hat{q}}^{(k)}=[\hat{q}_1^{(k)},\dots,\hat{q}_L^{(k)}]^\top\in\mathbb C^{L\times 1}\) contains the complex path coefficients associated with the \(L\) candidate SPs on the \(k\)-th subcarrier, which can be obtained as the closed-form least-squares solution for the $k$-th subcarrier as 
\begin{align}
\boldsymbol{\hat{q}}^{(k)}=\big((\boldsymbol\Upsilon^{(k)})^H\,\boldsymbol\Upsilon^{(k)}\big)^{-1} (\boldsymbol\Upsilon^{(k)})^H\,\boldsymbol s^{(k)}.
\end{align}

\subsubsection{Prior Statistics}
Each estimated SFP pair $(\hat{u}^{(\Theta_\ell)},\hat{v}^{(\Xi_\ell)})$ yields the
unit direction vector $\hat{\boldsymbol{a}}_\ell$ of SP $\ell$ in~\eqref{eq: revised AoA direction},
from which its AoA can be recovered via~\eqref{eq:recover} 
in the rotated frame $(\mathsf{X},\mathsf{Y},\mathsf{Z})$. 
To compare it with the prior
in Assumption~\ref{Assumption of AoA Prior}, which is defined in the initial frame
$(\mathsf{X}^{(0)},\mathsf{Y}^{(0)},\mathsf{Z}^{(0)})$, 
we rotate $\hat{\boldsymbol{a}}_\ell$ back to the initial frame as
$\hat{\boldsymbol{a}}^{(0)}_\ell=\boldsymbol{R}(\boldsymbol{\varphi}^\star)\,\hat{\boldsymbol{a}}_\ell$,
and recover the initial-frame AoA pair $(\hat{\theta}^{(0)}_\ell,\hat{\phi}^{(0)}_\ell)$
via~\eqref{eq:recover}.
Under the Gaussian model in Assumption~\ref{Assumption of AoA Prior}, the log-prior density function is
\begin{align}
\mathcal{D}(\Theta,\Xi)&\triangleq \log{p}\!\left(\{\hat{\theta}_\ell^{(0)}\!,\!\hat{\phi}_\ell^{(0)}\}_{\ell=1}^{L} \!\mid\! \Theta,\Xi\right)
\nonumber \\
=
-\sum_{\ell=1}^{L}&\!\bigg(
\frac{\big(\hat{\theta}^{(0)}_\ell\!-\!\mu_\ell\big)\!^2}{2\sigma_\ell^2}
\!+\!\frac{\big(\hat{\phi}^{(0)}_\ell\!-\!\xi_\ell\big)\!^2}{2\varsigma_\ell^2}
\!+\!\log(2\pi \sigma_\ell \varsigma_\ell)
\bigg)\!.
\label{eq:D_data}
\end{align}

By combining \eqref{eq:L_data} and \eqref{eq:D_data}, we formulate the MAP-based pairing rule as
\begin{align}\label{eq:map_matching_rule}
(\Theta^\star,\Xi^\star)
=
\arg\max_{\Theta,\Xi\in\mathcal S}
\Big(
\mathcal{L}(\Theta,\Xi)
+
\mathcal{D}(\Theta,\Xi)
\Big).
\end{align}
With the resulting optimal permutation pair \((\Theta^{\star},\Xi^{\star})\), the AoA pairs
\(\{(\hat\theta_{\ell}^{(0)}, \hat\phi_{\ell}^{(0)})\}_{\ell=1}^{L}\) can then be
recovered. However, directly evaluating \eqref{eq:map_matching_rule} requires
enumerating all \((L!)^2\) permutation pairs in \(\mathcal{S}\times\mathcal{S}\). For each pair,
evaluating \(\mathcal{L}(\Theta,\Xi)\) costs \(\mathcal{O}(KML^2)\) for the \(K\)
subcarrier-wise least-squares solves, while \(\mathcal{D}(\Theta,\Xi)\) costs
\(\mathcal{O}(L)\), yielding an overall complexity of
\(\mathcal{O}((L!)^2 \cdot KML^2)\).

\subsection{Low-Complexity MAP-Based AoA Estimation}\label{sec:proposed_low_complexity}
To avoid the exhaustive search required by~\eqref{eq:map_matching_rule}, the proposed algorithm first identifies $N$ promising SFP pairings and then associates the paired SFPs with the SP identities. This yields $N$ candidate permutation pairs, over which the MAP rule is evaluated. In the following, we explain the proposed two-step algorithm. 

\textbf{Step 1. Prior-Based SFP Pairing}:
Among the $L!$ one-to-one pairings between the SFPs in $\mathcal{U}$ and $\mathcal{V}$, this step identifies the $N$ most plausible ones. For each SFP pair $(i,j)$, where $i,j\in\{1,\ldots,L\}$, let $(\hat{\theta}^{(0)}_{i,j},\hat{\phi}^{(0)}_{i,j})$ denote the AoA recovered from $\hat{u}^{(i)}$ and $\hat{v}^{(j)}$. Since the SP identity associated with each pair is not yet known at this stage, we score the pair using the largest log-prior density of its recovered AoA over all SPs. This yields the score matrix $\boldsymbol{Q}\in\mathbb{R}^{L\times L}$ with entries
\([\boldsymbol{Q}]_{i,j} = \max_{\ell \in \{1,\ldots,L\}}
\log p_\ell\!\left(\hat{\theta}^{(0)}_{i,j}, \hat{\phi}^{(0)}_{i,j}\right)\),
where $p_\ell$ denotes the prior AoA density of SP $\ell$, as defined in Assumption \ref{Assumption of AoA Prior}. This score serves only to screen plausible SFP pairings. The one-to-one association between the recovered AoAs and the distinct SP identities is enforced in Step~2.

A one-to-one pairing is represented by a permutation \(\Omega\in\mathcal{S}\), where \(\Omega_i\) denotes the index of the SFP in 
\(\mathcal{V}\) paired with \(\hat{u}^{(i)}\). The resulting SFP pairs are therefore given by \(\{(\hat{u}^{(i)},\hat{v}^{(\Omega_i)})\}_{i=1}^{L}\). Because each SFP appears exactly once, a pairing selects \(L\) entries of \(\boldsymbol{Q}\), with one entry from each row and each column. We define its aggregate score as
\begin{align}\label{eq:pairing_score}
\mathcal{J}(\Omega)=\sum_{i=1}^{L}\big[\boldsymbol{Q}\big]_{i,\Omega_i}.
\end{align}
A larger score indicates that the AoAs recovered under the pairing are
collectively more consistent with the prior AoA statistics. We therefore
select the \(N\) highest-scoring pairings by solving
\begin{align}
\mathcal{A}^{\star}
\triangleq
\underset{
\substack{
\mathcal{A}\subseteq\mathcal{S}\\
|\mathcal{A}|=N
}}
{\arg\max}
\quad
\sum_{\Omega\in\mathcal{A}}
\mathcal{J}(\Omega),
\label{eq:problem_formulation_3}\tag{P3}
\end{align}
where
\(\mathcal{A}^{\star}
=
\{\Omega^{(n)}\}_{n=1}^{N}\)
and the selected permutations are indexed such that
\begin{align}
\mathcal{J}\!\left(\Omega^{(1)}\right)
\ge
\mathcal{J}\!\left(\Omega^{(2)}\right)
\ge
\cdots
\ge
\mathcal{J}\!\left(\Omega^{(N)}\right).
\end{align}
Problem~\ref{eq:problem_formulation_3} is an \(N\)-best linear assignment
problem, which can be efficiently solved using Murty's
algorithm~\cite{Murty1968}.

\textbf{Step 2. SP Association}: 
The $n$-th candidate pairing, represented by $\Omega^{(n)}$, consists of $L$ SFP pairs, with the $i$-th pair given by $(\hat{u}^{(i)},\hat{v}^{(\Omega^{(n)}_i)})$. Its corresponding AoA,
recovered in Step 1, is
\((\hat{\theta}^{(0)}_{i,\Omega^{(n)}_i},
\hat{\phi}^{(0)}_{i,\Omega^{(n)}_i})\). This step determines the SP identity of each pair, which was left unresolved
in Step~1, by comparing its recovered AoA with the prior AoA statistics of
every SP. Specifically, for the \(n\)-th candidate pairing, we compute another scoring matrix $\boldsymbol{C}^{(n)}\in\mathbb{R}^{L\times L}$, whose entries are given by
\(\big[\boldsymbol{C}^{(n)}\big]_{i,\ell}
=
\log p_\ell\!\big(
\hat{\theta}^{(0)}_{i,\Omega^{(n)}_i},
\hat{\phi}^{(0)}_{i,\Omega^{(n)}_i}
\big)\),
where each entry quantifies the consistency between the AoA recovered from
the \(i\)-th pair and the prior AoA statistics of SP~\(\ell\).

Under the assumption that the $L$ detected pairs originate from $L$ distinct SPs, each SP must be associated with exactly one pair, and vice versa.
An association is thus represented by the permutation $\Theta\in\mathcal{S}$ as introduced in
Sec.~\ref{sec:SFP_pairing}, where $\Theta_\ell$ denotes the index of the SFP pair associated with SP $\ell$. The permutation \(\Theta\) selects \(L\) entries
of \(\boldsymbol{C}^{(n)}\), with one entry from each row and each column.
The resulting association score is defined as
\begin{align}
\mathcal{K}^{(n)}(\Theta)
=
\sum_{\ell=1}^{L}
\big[\boldsymbol{C}^{(n)}\big]_{\Theta_\ell,\ell}.
\label{eq:assoc_score}
\end{align}
A larger score indicates that the recovered AoAs are collectively more
consistent with the prior statistics of their assigned SPs. The optimal
association for the \(n\)-th candidate pairing is thus obtained as
\begin{align}
\Theta^{(n)}
\triangleq
\underset{\Theta\in\mathcal{S}}{\arg\max}
\quad
\mathcal{K}^{(n)}(\Theta),
\label{eq:problem_formulation_4}\tag{P4}
\end{align}
which is a linear assignment problem and can be solved using the Hungarian
algorithm~\cite{Kuhn1955}. Because each pair is indexed by its \(\mathcal{U}\)-SFP,
\(\Theta^{(n)}_\ell\) gives the \(\mathcal{U}\)-index assigned to SP~\(\ell\),
while
\begin{align}
\Xi^{(n)}_\ell
=
\Omega^{(n)}_{\Theta^{(n)}_\ell}
\end{align}
gives the corresponding \(\mathcal{V}\)-index.

Repeating Step 2 for all permutations in \(\mathcal{A}^{\star}\) yields the candidate set \(\{(\Theta^{(n)},\Xi^{(n)})\}_{n=1}^N\). The final permutation pair is obtained by evaluating the MAP criterion in~\eqref{eq:map_matching_rule} only over this reduced candidate set. 

\begin{figure}[t]
  \centering
  \includegraphics[width=0.83\linewidth]{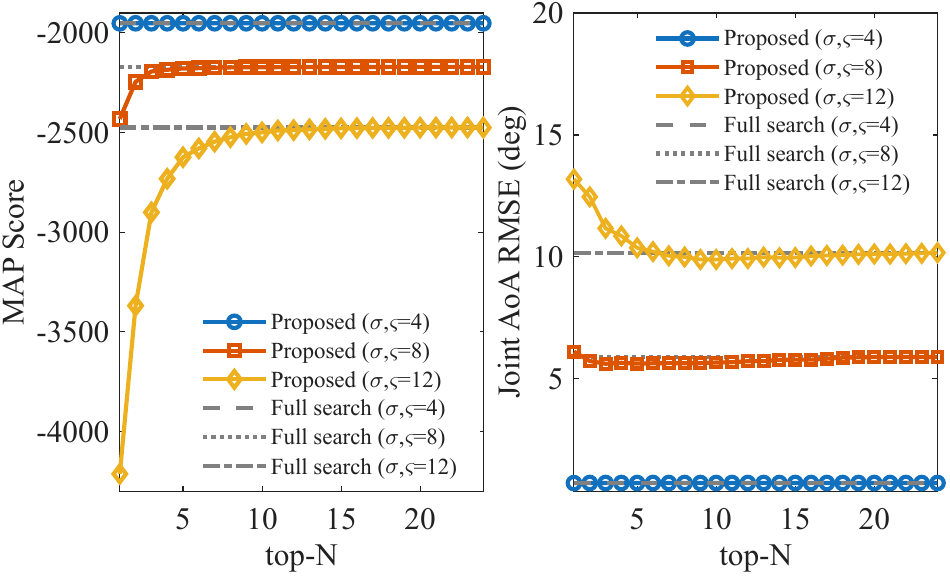}
  \vspace{-1mm}
\caption{MAP objective and joint AoA RMSE versus $N$.\vspace{-15pt}}
  \label{fig:topN}
\end{figure}

\begin{remark}[Effect of MA Orientation Control]\emph{Although the prior statistics provide only weak side information, they become more informative after orientation control separates the SPs, because a correctly matched SFP pair is more likely to yield an AoA consistent with the corresponding SP prior. It is numerically confirmed that even a small $N$ is sufficient to
include the correct permutation pair among the candidates, as observed in Fig.~\ref{fig:topN}.}
\end{remark}

\begin{remark}[Reduced Computation Complexity]\emph{For a prescribed $N\ll L!$, the two steps require $\mathcal{O}(NL^3)$, while the restricted MAP evaluation requires $\mathcal{O}(NKML^2)$. The proposed method thus replaces the exhaustive factorial search with a complexity that scales polynomially in $L$ and linearly in $N$. In this paper, we set $N=10$, which is found to be sufficient for reliable pairing.}
\end{remark}

\subsection{ToA Estimation via AoA-Matched Spatial Filtering}\label{sec:traj_refine}
To estimate each SP's ToA, we construct an AoA-matched spatial filter using the corresponding AoA estimates and apply it to the \(\mathsf{X}\)- and \(\mathsf{Z}\)-axis scan signals. For SP $\ell$, the spatial filter is defined as
$\boldsymbol{G}_{\ell}=\big[\boldsymbol{G}^{(1)}_{\ell},\boldsymbol{G}^{(2)}_{\ell}\big]
\in\mathbb{C}^{K\times 2M}$, where 
\begin{align}
\big[\boldsymbol{G}^{(n)}_\ell\big]_{k,m}
= \exp\!\left(j2\pi f_k \frac{\hat{\rho}_\ell(\boldsymbol{p}^{(n)}_m)}{c}\right),
\quad n\in\{1,2\}.
\end{align}
Here, $\hat{\rho}_\ell(\boldsymbol{p}^{(n)}_m)$ denotes the additional propagation distance in~\eqref{eq:rho_uv_def} evaluated with the estimated AoA pair. 

Let the received-signal matrix from the two scans be stacked across the
subcarriers as $\boldsymbol{S}=[\boldsymbol{s}^{(1)},\dots,
\boldsymbol{s}^{(K)}]^{\top}\in\mathbb{C}^{K\times 2M}$. The resultant spatially filtered signal for SP~\(\ell\) is obtained as
\begin{align}
\tilde{\boldsymbol{s}}_\ell
= \big(\boldsymbol{G}_\ell \odot \boldsymbol{S}\big)\,\boldsymbol{1}_{2M}
\in \mathbb{C}^{K},
\end{align}
where $\odot$ denotes the Hadamard product and $\boldsymbol{1}_{2M}$ is the all-ones vector. 

To characterize the effect of spatial filtering, the $k$-th element of
$\tilde{\boldsymbol{s}}_\ell$ can be expressed as
\begin{align}\label{eq:filtered_kth}
\big[\tilde{\boldsymbol{s}}_\ell\big]_k
&= \frac{P}{K}\sum_{u=1}^{L} \alpha_{u}
\exp\left(-j2\pi f_k \tau_{u}\right)  \nonumber\\
& \!\!\!\!\!\!\!\!\!\!\times \sum_{n=1}^{2}\sum_{m=1}^{M}
\!\exp\!\left(\!-j2\pi f_k\,
\frac{\rho_{u}(\boldsymbol{p}^{(n)}_m)\!-\!\hat{\rho}_{\ell}(\boldsymbol{p}^{(n)}_m)}{c}
\!\right) 
\!+\! \tilde{w}^{(k)}\!,
\end{align}
where $\tilde{w}^{(k)}$ is the filtered noise following 
$\mathcal{CN}(0, 2MPN_0/K)$. Under the
narrowband approximation and assuming accurate AoA estimates,
\eqref{eq:filtered_kth} can be decomposed into the
target-SP component ($u = \ell$) and the residual inter-SP interference as 
\begin{align}\label{eq:filtered_final}
\big[\tilde{\boldsymbol{s}}_\ell\big]_k
\approx& \frac{2MP}{K}\alpha_\ell\! \exp\!\left(-j2\pi f_k \tau_\ell\right)
\!+\! \frac{P}{K}\!\sum_{u \neq \ell} \alpha_u
\exp\!\left(-j2\pi f_k \tau_u\right) \nonumber\\
&\!\!\!\!\!\!\!\!\!\!\!\!\!\!\!\!\!\times\!\sum_{n=1}^{2}\!
\exp\!\left(j (M{\!-\!}1)\pi \frac{f_c\Delta\rho^{(n)}_{\ell,u}}{c}\!\right)
\!\frac{\!\sin\!\big(M\pi f_c \Delta\rho^{(n)}_{\ell,u}/c\big)}
{\!\sin\!\big(\pi f_c \Delta\rho^{(n)}_{\ell,u}/c\big)},
\end{align}
where the noise term is omitted for brevity. Here, $\Delta\rho^{(1)}_{\ell,u}$ and $\Delta\rho^{(2)}_{\ell,u}$ denote the per-step
difference between SPs $\ell$ and $u$ in the additional propagation distance in~\eqref{eq: propagation dist}, along the $\mathsf{X}$- and
$\mathsf{Z}$-axis scans, respectively, given as
\begin{align}
\Delta\rho^{(1)}_{\ell,u}
&= d\big(\sin\hat{\theta}_\ell\cos\hat{\phi}_\ell
-\sin\theta_u\cos\phi_u\big), \nonumber\\
\Delta\rho^{(2)}_{\ell,u}
&= d\big(\cos\hat{\theta}_\ell-\cos\theta_u\big). \label{eq:drho2}
\end{align}

\begin{figure}[t]
  \centering

  \begin{subfigure}[t]{0.85\linewidth}   
    \centering
   \includegraphics[width=1\linewidth]{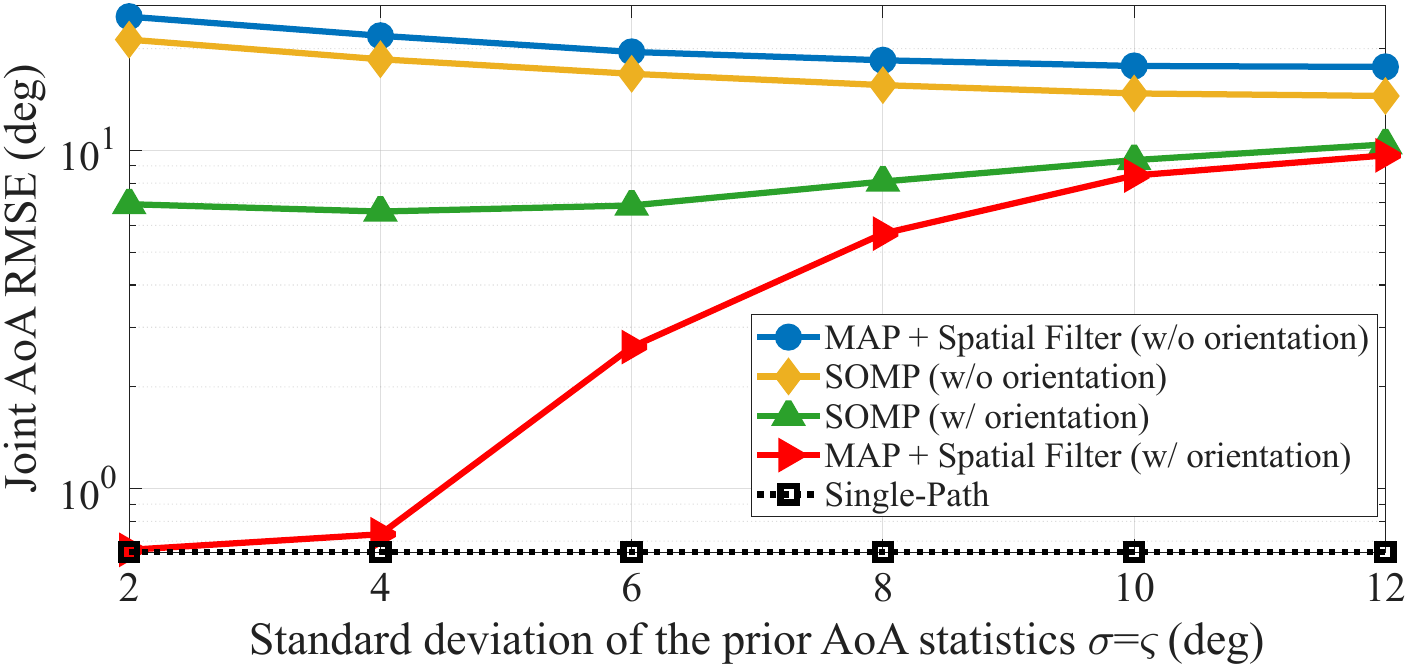}  
    \vspace{-5mm}
    \caption{}
    \label{fig:aoa_rmse}
  \end{subfigure}

  \begin{subfigure}[t]{0.85\linewidth}
    \centering
    \includegraphics[width=1\linewidth]{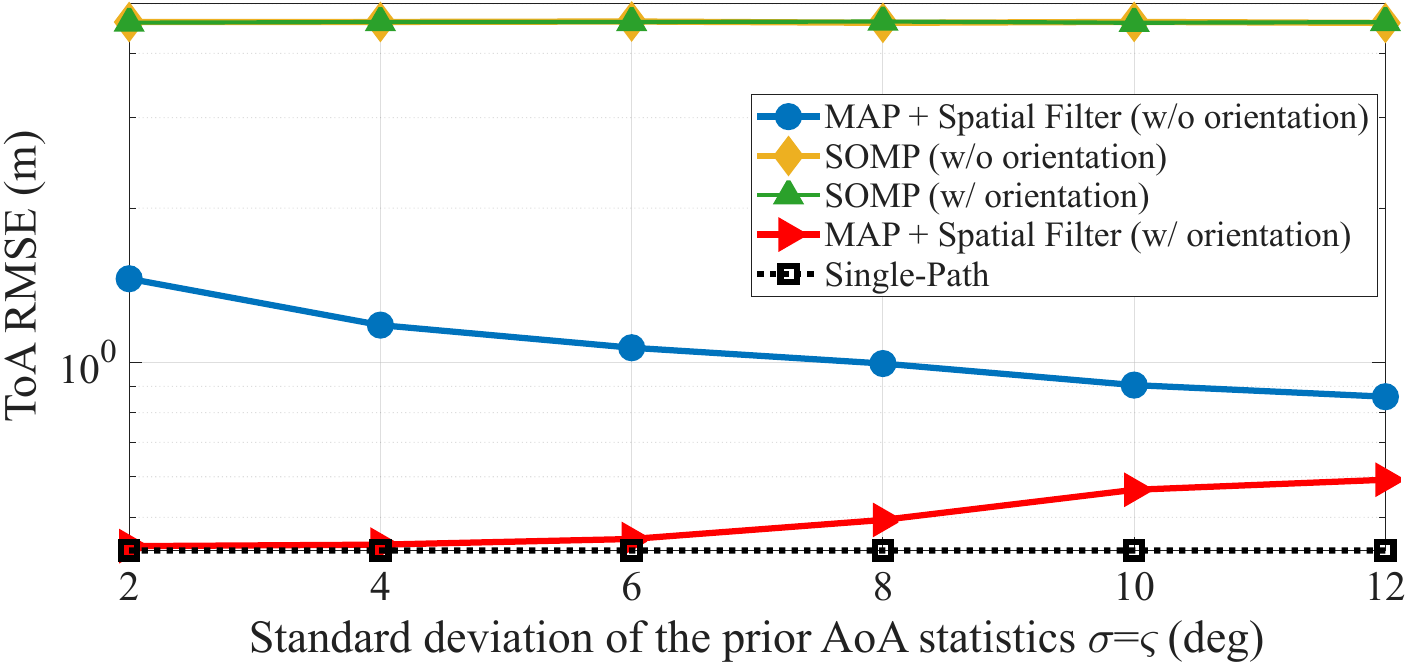}
    \vspace{-5mm}
    \caption{}
    \label{fig:toa_rmse}
  \end{subfigure}

  \vspace{-1mm}
  \caption{(a) AoA RMSE and (b) ToA RMSE versus the prior standard deviation.  }
  \label{fig:rmse}
  \vspace{-16pt}
\end{figure}

The first term in~\eqref{eq:filtered_final}, corresponding to the target
SP \(\ell\), coherently combines all $2M$ measurements and thus achieves an amplitude gain of \(2M\). On the other hand, the second term contains the components of the other SPs, which act as interference in estimating the target ToA $\tau_{\ell}$. For each interfering SP, the contribution from
the \(n\)-th scan is weighted by a Dirichlet kernel. Its magnitude is substantially reduced when its
per-step propagation-distance mismatch lies outside the main lobe, namely,
\begin{align}\label{eq:nulling_condition}
\big|\Delta\rho^{(n)}_{\ell,u}\big| \ge \frac{c}{Mf_c}.
\end{align}
From the Dirichlet-kernel form in~\eqref{eq:filtered_final} and the
main-lobe condition in~\eqref{eq:nulling_condition}, we obtain the following
corollary.

\begin{corollary}[Main-Lobe Interference Attenuation via Spatial Filtering]\label{prop:nulling}
\emph{Consider the spatial filter $\boldsymbol{G}_{\ell}$ designed for SP $\ell$. Under the main-lobe approximation, the contribution of SP~\(u\)
along the \(n\)-th scan lies outside the main lobe if \eqref{eq:nulling_condition} is satisfied. Consequently, the contributions
of all interfering SPs lie outside the main lobes of both scan responses if
\begin{align}\label{eq:M_condition}
M \ge \frac{c}{f_c\,\Delta\rho_{\min}},
\end{align}
where $\Delta\rho_{\min}\triangleq
\min_{n\in\{1,2\}}\ \min_{\ell\neq u}\ \big|\Delta\rho^{(n)}_{\ell,u}\big|$.
Under this condition, the main-lobe interference is attenuated, although
residual sidelobe components may remain. }
\end{corollary}

\begin{remark}[Effect of Orientation Control on Spatial Filtering]\emph{
The orientation control in Sec.~\ref{sec:opt} enlarges the minimum inter-SP separation, as exemplified in Remark~\ref{Remark1}, where the value of $\Delta\rho_{\min}$ increases from $0.0259d$ to $0.2412d$. This facilitates reducing the number of MA positions $M$ required for effective interference suppression. }
\end{remark}

With the other SPs sufficiently attenuated, the ToA of the target SP is encoded in the linear phase progression of \(\tilde{\boldsymbol{s}}_\ell\) across the \(K\) subcarriers. It is thus estimated as
\begin{align}
\hat{\tau}_\ell
=
\underset{\tau}{\arg\max}\;
\big|
\sum_{k=1}^{K}
\big[\tilde{\boldsymbol{s}}_\ell\big]_k
\exp\!\left(j2\pi f_k\tau\right)
\big|^2,
\label{eq:toa_estimation}
\end{align}
whose objective attains its maximum at \(\tau=\tau_\ell\). Repeating this procedure for
\(\ell\in\{1,\ldots,L\}\) yields the final set of estimated SP parameters,
 \(\{(\hat\tau_\ell,\hat\theta_\ell^{(0)}, \hat\phi_\ell^{(0)})\}_{\ell=1}^{L}\).

\section{Simulation Results}
This section evaluates the proposed framework for estimating the SP parameters \(\{(\hat\tau_\ell,\hat\theta_\ell^{(0)}, \hat\phi_\ell^{(0)})\}_{\ell=1}^{L}\) under different communication conditions, and further validates its real-world applicability through ray-traced environments.

\begin{figure}[t]
\centering
\includegraphics[width=0.85\linewidth]{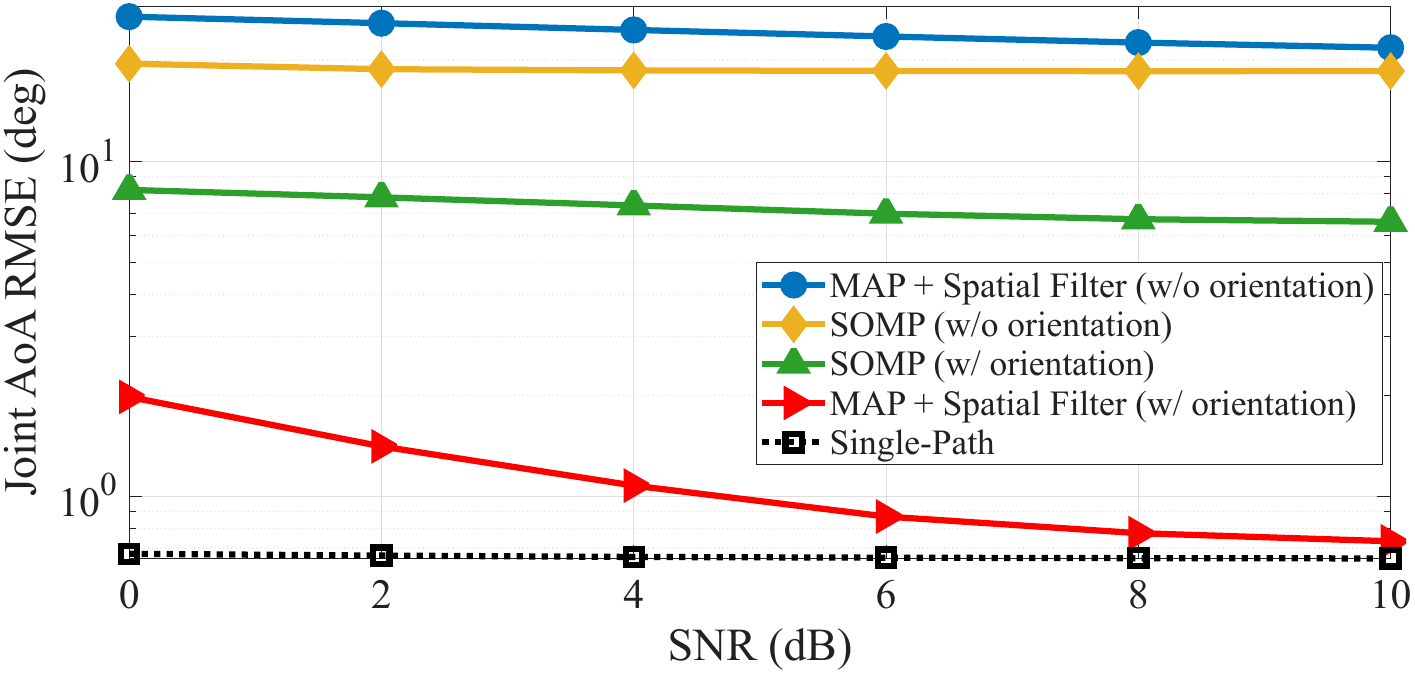}
\vspace{-3mm}
\caption{Joint AoA RMSE versus SNR. \vspace{-14pt}}
\label{fig:snr}
\end{figure}

\begin{figure}[t]
\centering
\includegraphics[width=0.85\linewidth]{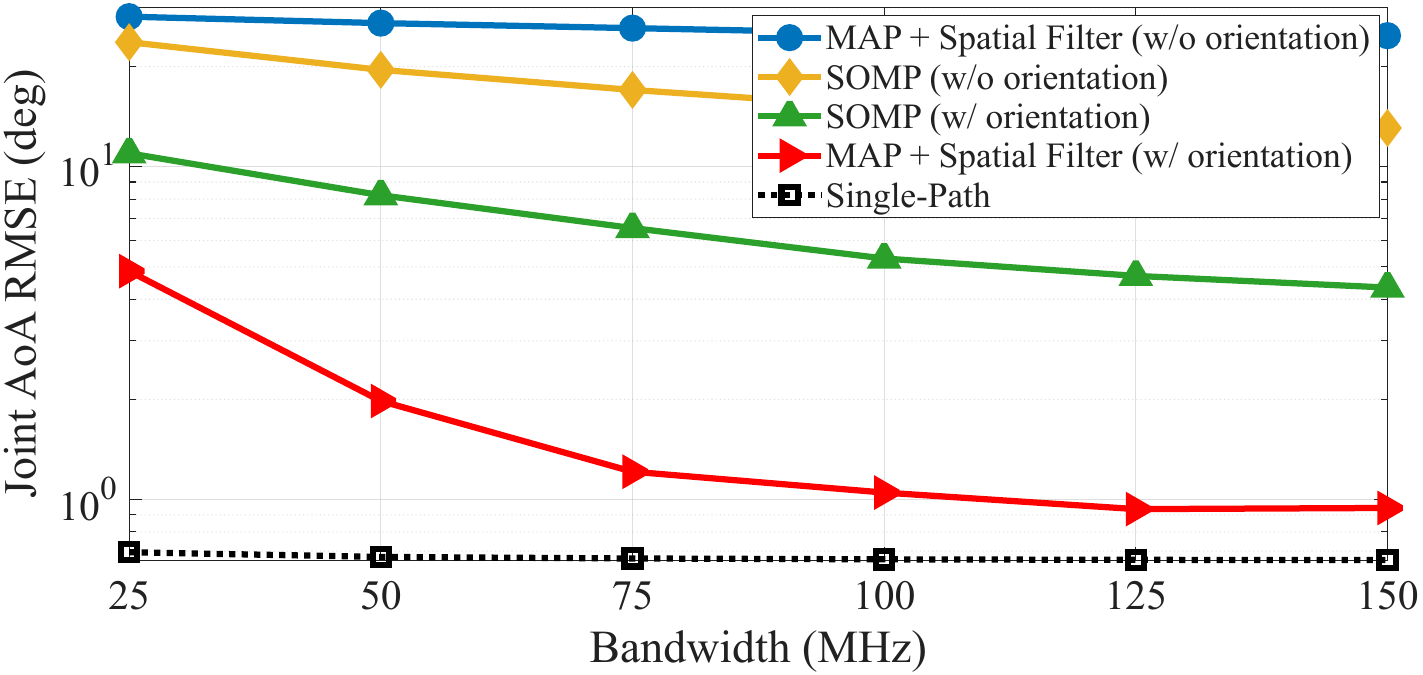}
\vspace{-2mm}
\caption{Joint AoA RMSE versus bandwidth. \vspace{-15pt} }
\label{fig:bandwidth}
\end{figure}

\subsection{Simulation Setup}\label{sec: setup}
Unless specified otherwise, the simulation parameters are set as follows.
We set $f_c=28$~GHz, $B_0=50$~MHz, and $K=64$ subcarriers, and assume that
the direct path is blocked so that all SPs are NLoS. A single MA collects $M=16$ measurements per scan with step size
$d=\lambda_c/2$, where $\lambda_c = c/f_c$ denotes the carrier
wavelength. All results are averaged over $5{,}000$ Monte Carlo trials. The AoAs of the $L=4$ SPs are drawn independently from their priors with means $\{\mu_\ell\}=\{115^\circ,98^\circ,51^\circ,51^\circ\}$ and
$\{\xi_\ell\}=\{56^\circ,115^\circ,40^\circ,121^\circ\}$ and standard deviations $\sigma_\ell=\varsigma_\ell=4^\circ$, where two SPs share the same elevation mean to make the pairing non-trivial, and the ToAs from $\tau_\ell\sim\mathcal{U}[\tau_0,\tau_0+1/B_0]$, so that the SPs are highly overlapped in both angle and delay at the nominal bandwidth. The AoA search uses a $101$-point grid over the SFP domain $[-1,1]$, the delay search in~\eqref{eq:toa_estimation} uses an oversampling rate of $4$, and the complex path gains follow the 3GPP path-loss model~\cite{TR38901}. Finally, the plate
orientation is optimized with the front-side tolerance $\epsilon=0.05$ using multi-start SQP with $100$ random initializations, and the proposed estimator retains the $N=10$ most plausible pairings. The AoA and ToA estimation performance is evaluated using the \emph{root mean square error} (RMSE), defined as
\begin{align}
\mathrm{RMSE}(\mathrm{AoA})
&=
\sqrt{\mathsf{E}\!\left[\left(\|\boldsymbol{\theta}-\hat{\boldsymbol{\theta}}\|_2^2
+\|\boldsymbol{\phi}-\hat{\boldsymbol{\phi}}\|_2^2\right)\!/L\right]}, \nonumber\\
\mathrm{RMSE}(\mathrm{ToA})
&=
\sqrt{\mathsf{E}\!\left[\|\boldsymbol{\tau}-\hat{\boldsymbol{\tau}}\|_2^2/L\right]}.
\end{align}

For AoA and ToA estimation, the proposed algorithm, referred to as
\textbf{MAP + Spatial Filter (w/ orientation)}, is compared with the following
four benchmarks:
\begin{itemize}[leftmargin=*]
    \item \textbf{MAP + Spatial Filter (w/o orientation)}: given an initial orientation, applies the same
    MAP-based pairing and spatial filtering as the proposed algorithm.
    \item \textbf{SOMP (w/o orientation)}: given an initial orientation, applies SOMP to estimate the AoAs
    over an angle dictionary and the ToAs over a delay
    dictionary~\cite{SCao2025}.
    \item \textbf{SOMP (w/ orientation)}: With the optimized plate orientation, applies the same SOMP-based
    estimation.
    \item \textbf{Single-Path}: serves as a lower bound for the single-SP case (\(L=1\)).
\end{itemize}
We adopt the SOMP-based scheme of~\cite{SCao2025} as our benchmark,
since it shares our far-field, wideband, single-MA setup.

\begin{figure}[!t]
  \centering

  \begin{subfigure}[t]{0.85\linewidth}
    \centering
    \includegraphics[width=\linewidth]{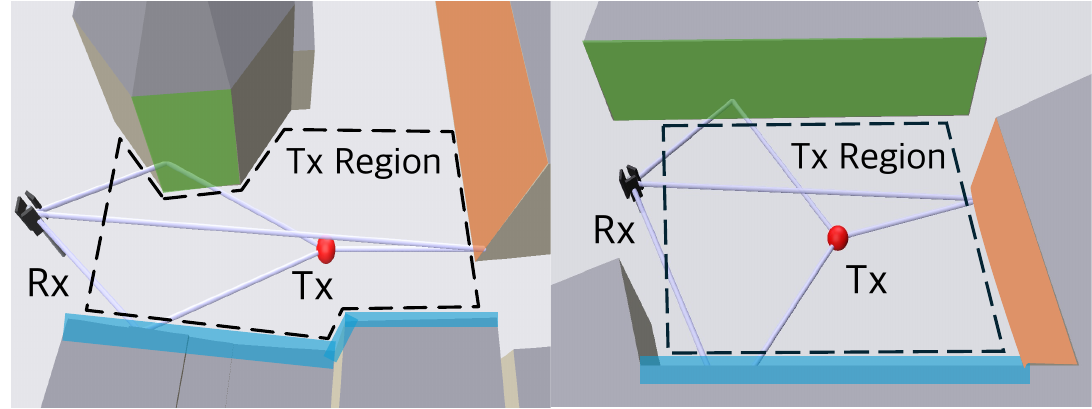}
    \caption{}
    \label{fig:subfig1}
  \end{subfigure}

  \begin{subfigure}[t]{0.85\linewidth}
    \centering
    \includegraphics[width=\linewidth]{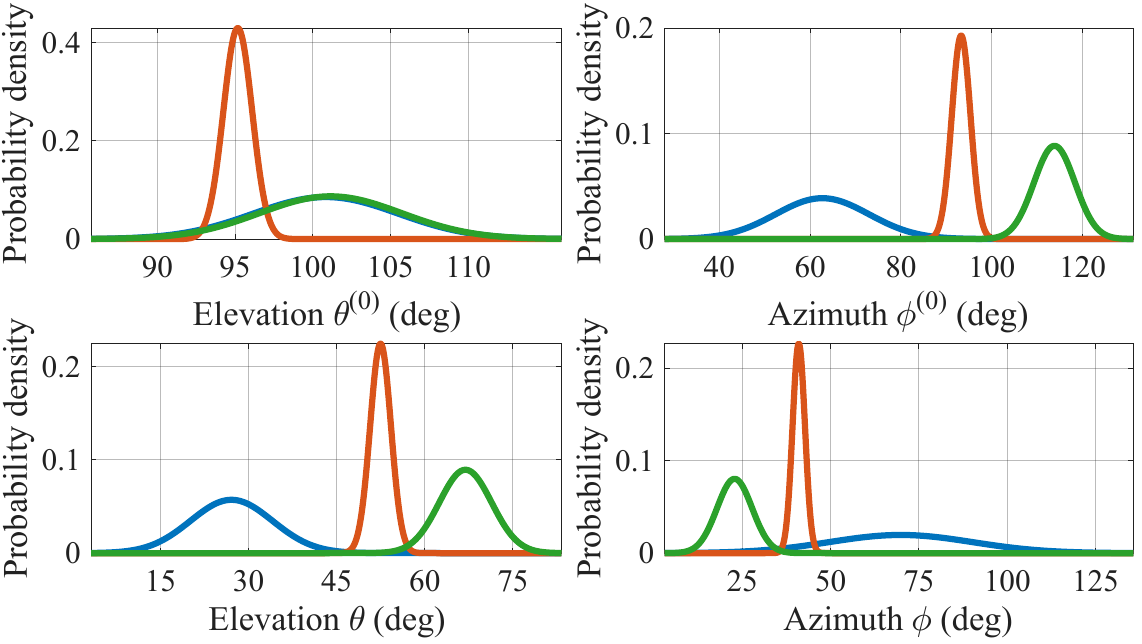}
     \vspace{-7mm}
    \caption{}
    \label{fig:subfig2}
  \end{subfigure}

  \begin{subfigure}[t]{0.85\linewidth}
    \centering
    \includegraphics[width=\linewidth]{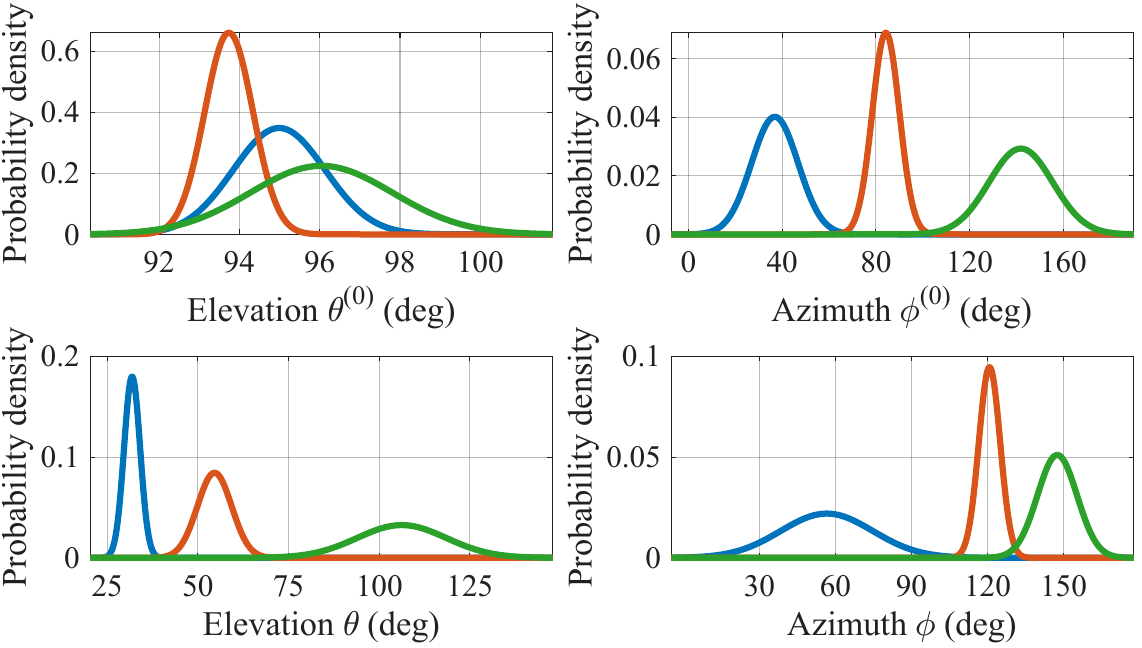}
     \vspace{-7mm}
    \caption{}
    \label{fig:subfig3}
  \end{subfigure}

\caption{(a) Multi-path environments of the Florence Duomo (left) and the Inha Aerospace Campus (right), and the SP AoA distributions in (b) and (c), respectively, shown in the initial (top) and optimally rotated (bottom) frames.}
\label{fig:real_inha_all}
\vspace{-6mm} 
\end{figure}

\subsection{Performance Evaluation}
Fig.~\ref{fig:rmse}(\subref{fig:aoa_rmse}) and~(\subref{fig:toa_rmse}) show the joint AoA RMSE and ToA RMSE, respectively, versus the standard deviations of the prior elevation and azimuth AoA statistics, say $\sigma$ and $\varsigma$, with $\sigma=\varsigma$. A larger standard deviation corresponds to less informative prior AoA statistics. Several interesting observations are made. First, the proposed algorithm outperforms all benchmarks for both AoA and ToA estimations and approaches the Single-Path bound. This gain arises because the orientation control leverages the prior AoA statistics to separate the SPs along the scan axes and facilitate unambiguous SFP pairing. Second, SOMP benefits only marginally from orientation control because it is not designed to exploit the prior AoA statistics when resolving highly correlated SPs. Third, as the standard deviation increases, the performance of the proposed algorithm degrades and its advantage over the benchmarks diminishes, which is consistent with the reduced informativeness of the prior AoA statistics. Finally, more accurate AoA estimates help construct the subsequent AoA-matched spatial filter in Sec.~\ref{sec:traj_refine}, thereby enabling more precise ToA estimation. Due to this close dependence and the page limit, the following results focus only on AoA-estimation performance.

\begin{figure}
  \centering
  \begin{subfigure}[t]{0.85\linewidth}
    \centering
    \includegraphics[width=\linewidth]{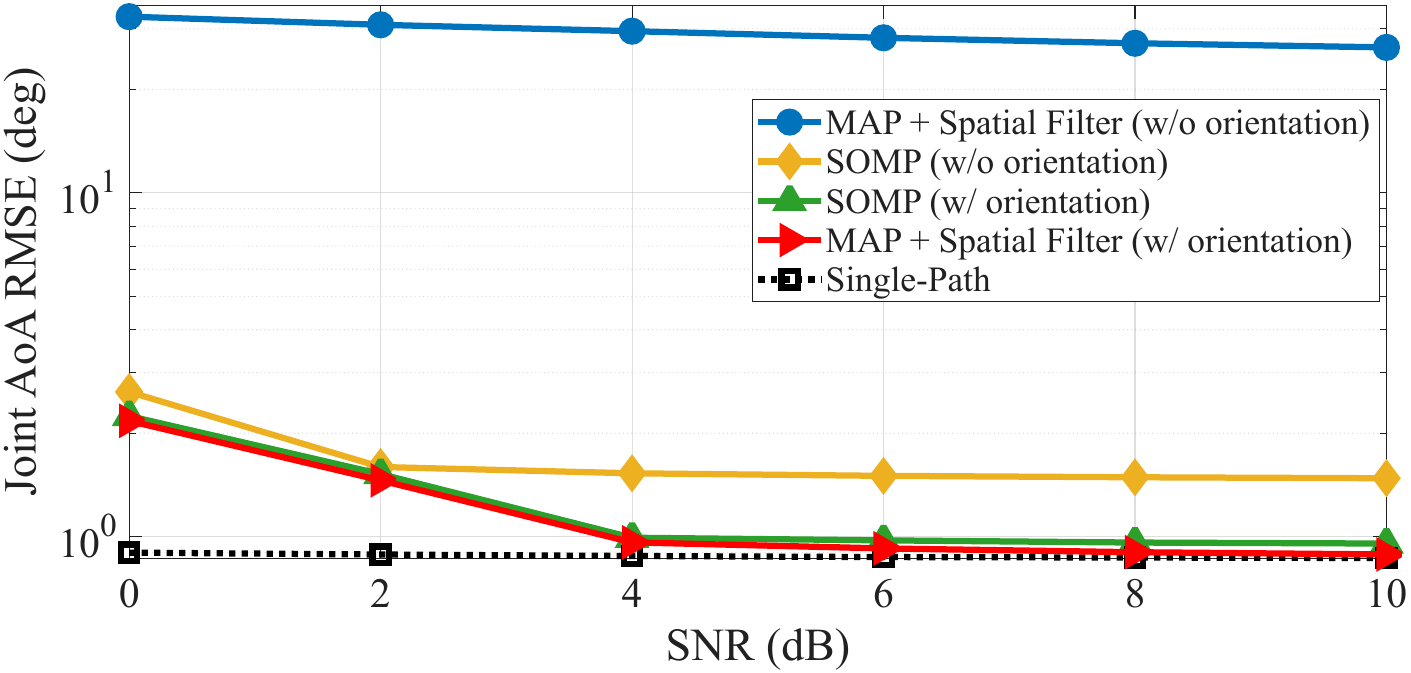}
    \caption{}
    \label{fig:real_florence}
  \end{subfigure}
  \begin{subfigure}[t]{0.85\linewidth}
    \centering
    \includegraphics[width=\linewidth]{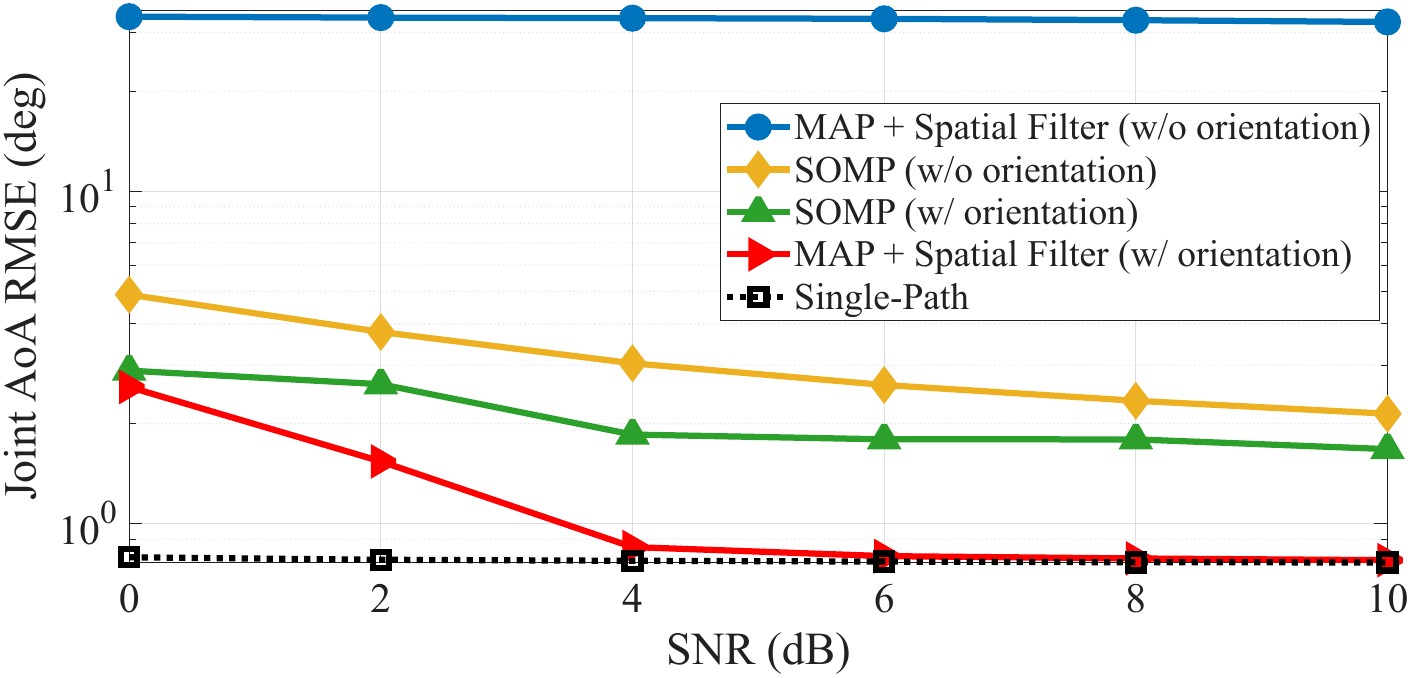}
    \caption{}
    \label{fig:real_inha}
  \end{subfigure}
  \caption{Joint AoA RMSE versus SNR in realistic ray-traced environments:
  (a) Florence Duomo and (b) Inha Aerospace Campus. }
  \label{fig:real_rmse}
  \vspace{-6mm} 
\end{figure}

Fig.~\ref{fig:snr} presents the joint AoA RMSE as a function of SNR. Except for the proposed algorithm, all benchmarks exhibit performance saturation and achieve only marginal gains at high SNR. This is because increasing SNR strengthens all SP components simultaneously, but does not improve their separability when they are highly correlated. On the other hand, the proposed algorithm consistently improves with SNR and nearly attains the Single-Path bound beyond $10$~dB. By controlling the plate orientation, the SP projections are more likely to be well separated along the scan axes, enabling more accurate AoA estimation with substantially reduced inter-SP interference. 

Fig.~\ref{fig:bandwidth} presents the effect of bandwidth on joint AoA RMSE.
Similar to the SNR counterpart in Fig.~\ref{fig:snr}, the proposed algorithm consistently outperforms the considered benchmarks. As recalled in Sec.~\ref{sec: setup}, all SPs' ToAs are generated within a single delay-resolution bin. In other words, increasing the bandwidth alone cannot effectively resolve the SP components in the delay domain. These results thus demonstrate the effectiveness of the proposed orientation control and the subsequent AoA-estimation algorithm for resolving closely spaced SPs.

We validate the proposed framework in two ray-traced environments emulating the Florence Duomo and the Inha Aerospace Campus. In both environments, the LoS path is blocked, and only single-bounce reflected paths are considered as SPs. The prior statistics in Table~\ref{tab:aoa_prior_stats} are extracted from $10{,}000$ realizations per environment with randomly placed transmitters, as
illustrated in Fig.~\ref{fig:real_inha_all}(\subref{fig:subfig1}). Fig.~\ref{fig:real_inha_all}(\subref{fig:subfig2}) and~(\subref{fig:subfig3}) show that the initially overlapped SPs become well separated on average after applying the optimized orientation. Fig.~\ref{fig:real_rmse} presents the resultant joint AoA RMSE versus SNR over $4{,}000$ realizations with \(L=3\), demonstrating that the proposed algorithm nearly attains the Single-Path bound.
Compared to the Florence Duomo, the performance gap between SOMP with orientation control and the proposed algorithm is more significant in the Inha campus, where nearby buildings generate reflected SPs with similar delays. Even in this harsh environment, the proposed algorithm effectively separates the SP components and accurately estimates their parameters, whereas the benchmark schemes fail to do so.   

\section{Conclusion}
This work has proposed a prior-guided MA control framework for agile multi-path sensing that exploits AoA statistics induced by the surrounding environment. Although these statistics provide only weak prior information, they play two important roles in the proposed framework. First, before MA scanning, the three-dimensional orientation of the movable plate is optimized to enhance the separation of multiple SP projections while controlling their dispersion, based on Fisher information analysis. Second, the MA performs only two linear scans, which yield complementary AoA-related parameters that should be paired to recover the elevation and azimuth AoAs of each SP. To enable reliable
pairing, the prior AoA statistics are incorporated into a MAP-based formulation,
together with a low-complexity candidate-selection algorithm. The recovered
AoAs are subsequently used to estimate the corresponding ToAs through
AoA-matched spatial filtering. With only one orientation adjustment and two
linear scans, the proposed framework accurately estimates the AoAs and ToAs of
multiple SPs while substantially reducing mechanical control overhead and
sensing latency. Numerical results under both stochastic channel models and
realistic ray-traced environments demonstrate that the proposed framework
outperforms the considered benchmarks and approaches the estimation performance
of the single-SP benchmark.

This work can be extended in several interesting directions. First, the estimated AoAs and ToAs can support localization, mapping, and tracking. Second, the framework can be extended to time-varying AoA statistics in dynamic and non-terrestrial environments. Finally, a closed-loop design can update the AoA priors from new SP estimates and adapt the MA control to environmental changes.

\bibliographystyle{IEEEtran}
\bibliography{reference/ieeeabrv,reference/reference}

\end{document}